\documentclass[10pt,english,journal]{IEEEtran}

\usepackage{iftex}
\ifXeTeX
  \usepackage{amsmath,amssymb,amsfonts,mathtools}
  \usepackage{amsthm}
  \usepackage{newtxmath}
  \usepackage{fontspec}
\else
  \usepackage[T1]{fontenc}
  \usepackage{amsmath,amssymb,amsfonts,mathtools}
  \usepackage{amsthm}
  \usepackage{newtxtext,newtxmath}
\fi

\usepackage{bm}
\usepackage{booktabs}
\usepackage{algorithm}
\usepackage[noend]{algpseudocode}
\usepackage{graphicx}
\usepackage{microtype}
\usepackage{url}
\usepackage{balance}
\usepackage{cite}
\usepackage{enumitem}
\usepackage{siunitx}
\usepackage{verbatim}
\usepackage{tabularx}   
\usepackage{array}      
\usepackage{adjustbox}
\usepackage{comment}
\usepackage{hyphenat}
\usepackage{xcolor}
\usepackage{tikz}
\usetikzlibrary{positioning,arrows.meta,fit,backgrounds}
\usepackage{placeins}
\usepackage[hidelinks]{hyperref}
\usepackage[nameinlink]{cleveref}   

\newcommand{\E}{\mathbb{E}}

\DeclareMathOperator*{\argmax}{arg\,max}

\newcommand{\VoI}{\mathrm{VoI}}

\allowdisplaybreaks 

\theoremstyle{plain}
\newtheorem{theorem}{Theorem}
\newtheorem{prop}{Proposition}

\theoremstyle{definition}

\theoremstyle{remark}

\makeatletter
\long\def\@makecaption#1#2{%
\ifx\@captype\@IEEEtablestring%
\footnotesize\bgroup\par\centering\@IEEEtabletopskipstrut{\normalfont\footnotesize #1}\\{\normalfont\footnotesize #2}\par\addvspace{0.5\baselineskip}\egroup%
\@IEEEtablecaptionsepspace
\else
\@IEEEfigurecaptionsepspace
\setbox\@tempboxa\hbox{\normalfont\footnotesize {#1.}\nobreakspace\nobreakspace #2}%
\ifdim \wd\@tempboxa >\hsize%
\setbox\@tempboxa\hbox{\normalfont\footnotesize {#1.}\nobreakspace\nobreakspace}%
\parbox[t]{\hsize}{\normalfont\footnotesize\noindent\unhbox\@tempboxa#2}%
\else%
\ifCLASSOPTIONconference \hbox to\hsize{\normalfont\footnotesize\hfil\box\@tempboxa\hfil}%
\else \hbox to\hsize{\normalfont\footnotesize\box\@tempboxa\hfil}%
\fi\fi\fi}
\makeatother

\graphicspath{{figures/}}
\usepackage{multirow}
\newcommand{\SlackHi}{95.8}
\newcommand{\SlackLo}{95.2}
\newcommand{\BindExactHi}{76}
\newcommand{\BindExactLo}{84}
\newcommand{\BindGapHi}{5.3}
\newcommand{\BindGapLo}{2.9}
\newcommand{\ZeroValue}{38}
\newcommand{\WithinOne}{21}
\newcommand{\WithinTwo}{35}
\newcommand{\ConfMinRate}{98.4}
\newcommand{\IntentExcHi}{1.6}
\newcommand{\IntentExcLo}{4.9}
\newcommand{\LLMAgreeLo}{24}
\newcommand{\LLMAgreeHi}{34}
\newcommand{\LLMGenieGap}{12.6 [4.9, 20.8]}
\newcommand{\LLMRatioRange}{3.2--3.5}
\newcommand{\LLMdTQwen}{13.8 [0.8, 26.9]}
\newcommand{\LLMdTLlama}{14.2 [0.8, 28.2]}
\newcommand{\LLMdTPlanner}{$-$3.7 [$-$14.1, 7.0]}
\newcommand{\LLMInterQwen}{17.5 [0.1, 34.6]}
\newcommand{\QSevenText}{With the stronger Qwen2.5-7B (30 seeds), the difference was 5.2 [$-$11.2, 21.5] slots, not significant.}

\begin{document}

\title{Resource-Efficient Semantic Communication\\ for Heterogeneous Agentic Teams}

\author{
Farhad~Rezazadeh,~\IEEEmembership{Member,~IEEE},~Hatim~Chergui,~\IEEEmembership{Senior~Member,~IEEE},~Lingjia~Liu,~\IEEEmembership{Fellow,~IEEE},\\~and~Merouane~Debbah,~\IEEEmembership{Fellow,~IEEE}
\thanks{F. Rezazadeh is with the Technical University of Catalonia (UPC) and BrainOmega, 08028 Barcelona, Spain (e-mail: farhad.rezazadeh@upc.edu).}
\thanks{H. Chergui is with i2CAT Foundation, 08034 Barcelona, Spain (e-mail: hatim.chergui@i2cat.net).}
\thanks{L. Liu is with the Virginia Tech, 24061 Blacksburg, USA (e-mail: ljliu@vt.edu).}
\thanks{M. Debbah is with the Khalifa University of Science and Technology, 127788 Abu Dhabi, UAE (e-mail: merouane.debbah@ku.ac.ae).}
}

\markboth{Manuscript submitted for review}{Resource-Efficient Semantic Communication for Heterogeneous Agentic Teams}

\maketitle
\raggedbottom
\thispagestyle{empty}
\pagestyle{empty}

\begin{abstract}
Teams of autonomous agents, including large language model (LLM) agents, must
coordinate over scarce and unreliable wireless links. We propose goal-oriented semantic
communication (GOSC)\footnote{To support
reproducibility, the source code is publicly available for non-commercial use
at \url{https://github.com/frezazadeh/gosc-agentic-teams}.}, a closed-loop co-design that jointly decides what each agent sends,
when it sends it, and how reliably it is transmitted, based on each message's value to
the team task. An edge broadcast of the team's common knowledge closes the loop by
updating these values. We prove that a message is sent only if its value exceeds the
cost of delivering it and that more valuable messages receive more robust transmission
rates, and we show that the scheduler solves each scheduling step exactly whenever the radio budget
is not saturated, which held in 95\% of scheduling decisions. In search-and-rescue
missions validated on unseen scenarios, GOSC meets the same mission targets as carefully
tuned periodic semantic schemes with 1.2--8.5 times fewer uplink channel uses. In most
settings, this advantage persists with realistic packet overheads, reaching 16.6 times
fewer uplink channel uses and 13.8 times lower cost when downlink costs are included; in
the rescue task, it also persists when all agents share one uplink. Rough value estimates suffice, whereas values that ignore
message content can fail. With three different LLMs, GOSC uses 3.2--3.5 times fewer
channel uses, while completion-time gains depend on the model.
\end{abstract}


\begin{figure*}[t]
\centering
\begin{tikzpicture}[font=\footnotesize, >={Stealth[length=5pt]}, node distance=6mm and 8mm,
  blk/.style={draw, rounded corners=2pt, align=center, minimum height=8mm, minimum width=26mm, inner sep=3pt},
  core/.style={blk, fill=blue!8, draw=blue!60!black},
  gosc/.style={blk, fill=orange!10, draw=orange!70!black},
  key/.style={blk, fill=orange!25, draw=orange!80!black, very thick},
  net/.style={blk, fill=gray!12},
  lab/.style={font=\scriptsize, fill=white, inner sep=1pt}]
\node[blk] (sense) {Perception\\(on-board detector)};
\node[blk, right=of sense] (belief) {Belief\\$\ell_k=\ell^E+\delta_k$};
\node[core, right=of belief] (cog) {Cognitive core\\planner \emph{or} LLM};
\node[blk, right=of cog] (int) {Intention $g_k$\\(target / claim)};
\draw[->] (sense) -- (belief);
\draw[->] (belief) -- (cog);
\draw[->] (cog) -- (int);
\node[gosc, below=9mm of belief] (enc) {\textbf{Q1 (what)} semantic\\encoding + relevance filter};
\node[gosc, right=of enc] (val) {Message value $V$\\(pluggable, e.g., ToM, bits, AoI)};
\node[key, right=of val] (sch) {\textbf{Q2 (when) \& Q3 (how reliably)}\\send iff $V>\lambda c^\star(\bar\gamma)$,\\rate non-increasing in $V$};
\draw[->] (belief) -- node[lab, right] {evidence} (enc);
\draw[->] (int.south) -- ++(0,-5mm) -| node[lab, pos=0.25, above] {intention} (val.north);
\draw[->] (enc) -- (val);
\draw[->] (val) -- node[lab, above] {$v$} (sch);
\begin{scope}[on background layer]
\node[draw=black!45, dashed, rounded corners, fit=(sense)(int)(enc)(sch), inner sep=3mm] (agent) {};
\end{scope}
\node[anchor=south west, font=\footnotesize\itshape] at (agent.north west) {Agent $k$};
\node[net, below=15mm of val, minimum width=50mm] (edge) {Base station + edge coordinator\\fusion of $\ell^E$, declarations, broadcast};
\draw[->] (sch.south) |- node[lab, pos=0.3, right, align=left] {fading uplink, $W$ ch.\ uses/slot,\\price $\eta$, rate $R\in\mathcal{R}$} (edge.east);
\draw[->] (edge.west) -| node[lab, pos=0.7, left, align=right] {common knowledge\\closes the loop} (enc.south);
\node[net, left=10mm of edge.west, anchor=east, yshift=-11mm] (mates) {Teammates\\(planner or LLM)};
\draw[<->] (mates.east) -| ([xshift=-6mm]edge.south);
\end{tikzpicture}
\caption{GOSC as a closed-loop co-design of what, when and how reliably. Each agent's
cognitive core (a planner or an LLM) produces beliefs and intentions; the semantic
representation and relevance filter determine the candidate messages (Q1); a pluggable
value function prices them; one value-priced selection decides whether each is sent now
(Q2) and at which rate (Q3) under the per-slot budget and the price $\lambda$
(Theorem~\ref{thm:structure}); the edge broadcast of the fused common knowledge tells every
sender what its teammates already know, closing the loop (Section~\ref{sec:gosc}). ToM,
theory of mind; AoI, age of information.}
\label{fig:arch}
\end{figure*}
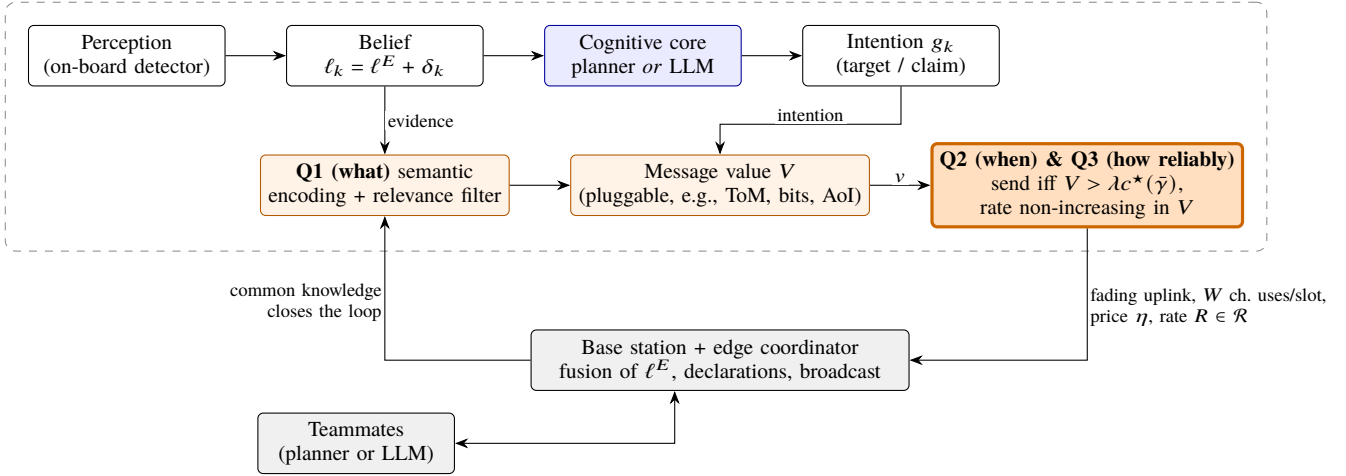

\section{Introduction}
\IEEEPARstart{A}{utonomous} agents increasingly solve tasks as teams, such as fleets of
unmanned aerial vehicles (UAVs), robots, and embodied agents built on large language
models (LLMs) \cite{saad2020vision,letaief2019roadmap,park2023generative,zhang2024coela, rezazadeh2024genonet}.
Coordination requires agents to share what they observe and what they intend, but in
the situations where teams matter most, such as disaster response, the wireless
capacity available for coordination is scarce, fading and shared with other traffic.
Classical communication design reproduces bits reliably irrespective of their use
\cite{shannon1949}; semantic and goal-oriented communication instead asks which
information serves the receiver's task
\cite{strinati2021beyond,gunduz2023beyond,kountouris2021semantics,chaccour2025less}.

For a team, this question becomes three decisions that every agent faces in every
slot.
\begin{enumerate}[label=\textbf{Q\arabic*}),leftmargin=*,itemsep=0pt]
\item \emph{What} should it send? Raw sensor data, verbose natural language as
exchanged by LLM agents \cite{wu2023autogen,li2023camel,hong2024metagpt}, or compact
semantic updates?
\item \emph{When} should it send? Periodically, or only when a message is worth its
radio cost?
\item \emph{How reliably} should the network deliver it over a fading link?
\end{enumerate}
Existing designs answer these questions largely one at a time. Semantic and
task-oriented communication chooses the representation (Q1)
\cite{xie2021deepsc,shao2022ib,sun2025importance}, age- and value-of-information
scheduling chooses the timing (Q2) \cite{kaul2012aoi,holm2023goal,wu2024waypoint}, and
importance-aware unequal error protection chooses the reliability (Q3)
\cite{cao2025mimo,kim2025uep}. In multi-agent coordination the three are coupled.
Whether a message is worth sending depends on what the teammates already know, which
changes with every delivered message; whether it is worth its cost depends on the
channel uses needed to deliver it reliably; and how much protection it deserves
depends on how much it is worth.

This paper co-designs the three decisions in closed loop under a task-level objective.
In the proposed goal-oriented semantic communication (GOSC) framework
(Figure~\ref{fig:arch}), a lossless semantic representation determines \emph{what} is
sent; a value-versus-cost trigger determines \emph{when}, sending a message only if its
value to the team exceeds the price of delivering it; and joint value-aware rate
selection determines \emph{how reliably}, protecting more valuable messages more
strongly. The edge coordinator's broadcast of the fused common knowledge closes the
loop; it tells every sender what its teammates already know, and hence what is still
worth saying.

A design built on message values raises a question that has received little attention.
\emph{How accurate does semantic valuation need to be for resource-efficient agent
communication?} We answer it analytically and experimentally. A value error can change
a decision only for messages whose value lies close to the price of their delivery, and
the loss it causes is bounded by the price of the channel uses whose allocation changes,
not by the value itself. In missions, GOSC kept a significant channel-use advantage over tuned periodic schedules in 20 of 24 combinations of setting and value estimate whose rank correlation with the exact value was at least 0.89, including typical errors of a factor of seven and values known only to their order of magnitude; only content-independent values lost the advantage in the search task. Precise semantic valuation is
therefore not necessary; coarse task relevance combined with communication-aware
scheduling suffices, which also explains why a theory-of-mind value showed no
advantage over simpler values.

Our contributions are as follows.
\begin{itemize}
\item \textbf{Closed-loop co-design of what, when and how reliably.} We formulate team
coordination over a block-fading uplink as the joint minimization of mission time and
radio cost and propose GOSC, in which one decision rule, expected task value minus the
price of the channel uses, selects the messages, triggers their transmission and sets
their rates, and the common-knowledge broadcast closes the loop. The communication
layer is training-free and independent of the agents' cognitive core, which can be a
planner or an LLM.
\item \textbf{Structure and exactness of the scheduler.} A message is sent if and only if
its value exceeds the price of a successful delivery; at this margin it uses the
goodput-maximizing rate, and the more its value exceeds the margin, the more robust its
rate (Theorem~\ref{thm:structure}). The Lagrangian relaxation of the per-slot scheduling
problem is exact whenever the price keeps the budget slack and otherwise loses at most
one message's net value (Theorem~\ref{thm:exact}). The price kept the budget slack in
\SlackLo--\SlackHi\% of the logged instances, which explains why the heuristic was
exact in 99\% of them.
\item \textbf{Resource efficiency against tuned semantic baselines.} At operating points
frozen and re-evaluated on fresh seeds, GOSC meets common mission targets with
1.2--8.5 times fewer uplink channel uses than the most frugal tuned semantic periodic
schedules we could construct, including ones that filter or deliberately suppress
updates. The advantage survives per-packet overheads of 16--64 bits in five of six settings, growing with the overhead in the rescue task; it holds for the total uplink-plus-downlink cost including the common-knowledge broadcast; and it persists in the rescue task when all agents compete for a shared network-wide uplink, whereas in the search task it is significant only with the largest shared budget.
\item \textbf{How accurate must semantic values be?} We bound the effect of value errors
(Proposition~\ref{prop:robust}) and show in missions that GOSC's efficiency survives
value errors of a typical factor of seven but not content-independent values; semantic
valuation must be coarsely right, not precise, and a decision-aware theory-of-mind value
brings no consistent gain over throughput or age-of-information values.
\item \textbf{Reproducible evaluation.} Two tasks, finite-blocklength errors, a lossy and
delayed downlink, per-packet overheads, a shared uplink, a proof-of-concept with
heterogeneous LLM--planner teams (Qwen2.5-3B, Llama-3.2-3B and Qwen2.5-7B), paired
statistics, and hyper-parameters frozen on seeds disjoint from the evaluation seeds;
the code, configurations and per-mission results will be released publicly.
\end{itemize}

\section{Related Work and Positioning}\label{sec:related}
\subsubsection*{What to send}
Deep joint source--channel coding \cite{bourtsoulatze2019deep} and DeepSC
\cite{xie2021deepsc} transmit meaning rather than bits; task-oriented designs optimize
for a receiver's inference task \cite{shao2022ib,xie2022taskmulti}, and surveys identify
multi-agent collaboration as a key application
\cite{gunduz2023beyond,strinati2021beyond,kountouris2021semantics,chaccour2025less,uysal2022semantic}.
Importance-aware designs select the most important learned features under a performance
target \cite{sun2025importance}; with several transmitters, feature dimensions, their
precision and resource blocks can be chosen jointly under a per-slot budget
\cite{lei2026icotasc}. Here the importance is computed offline, or per input, by a
trained model for the inference task of a single receiver.

\subsubsection*{When to send}
Age of information (AoI) and its variants schedule updates by freshness
\cite{kaul2012aoi,kadota2018scheduling,maatouk2020aoii}; query AoI
\cite{chiariotti2022qaoi} and application-aware value of information \cite{holm2023goal}
schedule sensor updates by when and how receivers use them. Event-triggered sampling
sends only when the change of a signal justifies it \cite{astrom2002riemann}; in
networked control, the value of a packet is its effect on the control cost
\cite{soleymani2022voi}, and goal-oriented robotic waypoint transmission combines value
and age of information with proactive repetition \cite{wu2024waypoint}. Theory of mind
(ToM), attributing beliefs and intentions to others \cite{premack1978tom,baker2017tom},
lets ToM2C \cite{wang2022tom2c} learn when and with whom to share intentions, and
Thomas \emph{et al.} \cite{thomas2023pragmatic} build a pragmatic semantic link between
two agents; learned multi-agent communication
\cite{foerster2016learning,sukhbaatar2016commnet,das2019tarmac,kim2019schednet,tung2021effective}
requires end-to-end training.

\subsubsection*{How reliably to deliver}
Importance-aware unequal error protection allocates more power to important features
\cite{cao2025mimo} or protects them with stronger channel codes to meet learned
reliability targets \cite{kim2025uep}; rate--distortion optimized streaming schedules
interdependent media packets \cite{chou2006rd}. These designs fix what is sent and
decide its protection.

\subsubsection*{Positioning}
Table~\ref{tab:position} organizes the closest works by the three decisions and by the
loop that couples them; it compares problem settings and mechanisms.
Each decision has strong precedents, namely importance-aware feature selection for \emph{what}
\cite{sun2025importance,lei2026icotasc}, ToM-based and value- or age-based gating for
\emph{when} \cite{wang2022tom2c,wu2024waypoint}, and importance-aware power, coding and
precision for \emph{how reliably} \cite{cao2025mimo,kim2025uep,lei2026icotasc}. What is
new in GOSC is their closed-loop co-design for a team with a task-level objective. One
decision rule chooses the message, whether to send it now, and its rate, by one
criterion, namely the expected value to the team minus the price of reliable delivery. The value is
computed against the teammates' current knowledge, which the edge broadcast keeps
common, so every delivery changes what is worth sending next. The co-design has a
structure that none of the separate designs has, with a send threshold set by the price of a
successful delivery, a rate that moves from goodput-optimal to most robust as the value
grows (Theorem~\ref{thm:structure}), and a scheduler that is exact whenever the price
keeps the budget slack (Theorem~\ref{thm:exact}). It also answers a question the
separate designs do not raise, namely how accurate the value must be (Section~\ref{sec:acc}).
LLM-based agents coordinate through natural language
\cite{park2023generative,li2023camel,wu2023autogen,hong2024metagpt,zhang2024coela}; GOSC
lets them coordinate through typed semantics at a fraction of the cost. Because the
published systems differ in task, receivers and channel model, we evaluate their
\emph{mechanisms} (ToM-based valuation, age-of-information prioritization, and a
throughput value) within one architecture, against periodic semantic schedules tuned
to be as frugal as possible (Section~\ref{sec:results}).

\begin{table*}[t]
\centering
\caption{Positioning by the three coupled decisions (what, when, how reliably) and the
loop that GOSC closes; problem setting and mechanism (\checkmark,
addressed; (\checkmark), partly; ---, not addressed or not modelled; Rx, receiver;
ToM, theory of mind; VoI/AoI, value/age of information; LoS, line of sight; AWGN, additive
white Gaussian noise; MIMO, multiple-input multiple-output; FBL, finite blocklength)}
\label{tab:position}
\setlength{\tabcolsep}{3.5pt}
\resizebox{\textwidth}{!}{%
\begin{tabular}{@{}lcccccccc@{}}
\toprule
 & ToM2C & Thomas \emph{et al.} & Wu \emph{et al.} & Sun \emph{et al.} & Cao \emph{et al.} & Kim \emph{et al.} & Lei \emph{et al.} & \\
 & \cite{wang2022tom2c} & \cite{thomas2023pragmatic} & \cite{wu2024waypoint} & \cite{sun2025importance} & \cite{cao2025mimo} & \cite{kim2025uep} & \cite{lei2026icotasc} & \textbf{GOSC}\\
\midrule
\multicolumn{9}{@{}l}{\emph{Q1 What (semantic message selection)}}\\
\quad Content chosen by task relevance & whom to tell & learned content & --- & top features & --- & --- & feature dimensions & messages, lossless records\\
\multicolumn{9}{@{}l}{\emph{Q2 When (event triggering)}}\\
\quad Sends only when worth its cost & learned gate & --- & VoI/AoI priority & --- & --- & --- & --- & value $>$ price of delivery (Thm.~\ref{thm:structure})\\
\multicolumn{9}{@{}l}{\emph{Q3 How reliably (physical-layer protection)}}\\
\quad Per-message reliability & --- & learned codec & equal repetition & --- & power & code rate & precision & rate\\
\quad Protection grows with value & --- & --- & --- & --- & \checkmark{} (closed form) & \checkmark{} (learned targets) & --- & \checkmark{} (Thm.~\ref{thm:structure})\\
\multicolumn{9}{@{}l}{\emph{Co-design and closed loop}}\\
\quad Q1--Q3 decided jointly & --- & (\checkmark) Q1+Q3 & --- & --- & --- & --- & (\checkmark) Q1+Q3 & \checkmark{} (one rule)\\
\quad Value tracks receivers' current knowledge & \checkmark & \checkmark & \checkmark & --- & --- & --- & --- & \checkmark{} (common knowledge)\\
\quad Several receivers that act on messages & \checkmark & --- (2 agents) & --- & --- & --- & --- & --- (1 fusion Rx) & \checkmark\\
\quad Exactness or structure of the decision & --- & --- & --- & --- & closed form & optimality analysis & --- & Thms.~\ref{thm:structure}, \ref{thm:exact}; Prop.~\ref{prop:robust}\\
\multicolumn{9}{@{}l}{\emph{Network and agents}}\\
\quad Resource budget & --- & --- & --- & --- & power & --- & resource blocks & per-agent or shared channel uses\\
\quad Channel and packet errors & ideal & Rayleigh & Rayleigh, LoS & AWGN, Rayleigh & MIMO Rayleigh & Rayleigh & no error model & Rayleigh, FBL, overhead\\
\quad Training-free communication layer & --- & --- & (\checkmark) ordering & (\checkmark) selector & (\checkmark) allocation & (\checkmark) coding & (\checkmark) scheduler & \checkmark\\
\quad Heterogeneous LLM--planner agents & --- & --- & --- & --- & --- & --- & --- & \checkmark\\
\bottomrule
\end{tabular}}
\end{table*}

\section{System Model and Problem Formulation}\label{sec:model}

\subsection{Tasks, Perception and Beliefs}
A disaster area is a grid of $N\times N$ cells with $V$ survivors at unknown cells. $K$
agents launch next to a base station (BS) with an edge coordinator. In slot $t$, agent
$k$ at $p_k(t)$ senses the $(2r+1)^2$ footprint $\mathcal{F}(p_k)$ with detection and
false-alarm probabilities $p_d$ and $p_{fa}$, independently per cell. The evidence
record $o_k^t$ adds the log-likelihood ratio (LLR) $\lambda^+=\ln(p_d/p_{fa})$ to
detected and $\lambda^-=\ln((1-p_d)/(1-p_{fa}))$ to undetected footprint cells of a
log-odds belief map; Bayesian fusion is additive. A survivor is \emph{declared} when
the edge posterior reaches $\tau$ or an agent whose own posterior crossed $\tau$
confirms it.

\emph{Search task.} The mission ends when all survivors are declared; the metric is
the completion time $T$ (capped at $T_{\max}$). \emph{Rescue task.} Each survivor also
has a deadline $d_s\sim\mathcal{U}[60,180]$ slots; once located, one agent must stay at
it for $S_r=8$ slots before $d_s$, and the metric is the fraction of survivors saved.
An agent's intention to rescue a survivor is a \emph{claim}; unknown claims cause
duplicated travel or neglected survivors.

\subsection{Agents}
Each agent follows a belief--desire--intention architecture \cite{rao1995bdi}. Its
\emph{planner core} selects the target
\begin{equation}\label{eq:score}
g_k=\argmax_{c} S_k(c),\;\; S_k(c)=\frac{M_k(c)}{1+\beta\|c-p_k\|_1}\prod_{q\in\mathcal{D}_k}\rho(c,q),
\end{equation}
where $M_k(c)$ is the belief mass in the footprint centred on $c$, $\beta$ discounts
travel, and $\rho(c,q)=1-e^{-\|c-q\|^2/2\varsigma^2}$ repels from the known intentions of
higher-priority teammates and the known positions of the others ($\mathcal{D}_k$); a
hysteresis factor $h$ avoids oscillations. In the rescue task, an agent takes the most
urgent feasible task that is not known to be claimed and for which no free teammate is
known to be closer, and keeps it while it remains feasible. Agents move one cell per
slot. In Section~\ref{sec:llm}, an LLM core replaces the planner core for some agents;
the communication layer only uses beliefs and intentions (Figure~\ref{fig:arch}).

\subsection{Radio Resources, Uplink and Downlink}\label{sec:radio}
\emph{Resources.} By default, each agent has $W$ channel uses per 1\,s slot on orthogonal
resources. In 5G New Radio (NR) numerology 0, a resource block over a 1\,ms slot holds 168 resource
elements \cite{3gpp38211}, so the default $W=100$ is about 0.6 resource-block slots per
agent and second, that is, a thin coordination sub-channel in a cell loaded by other traffic. In
Section~\ref{sec:shared}, all agents instead compete for a shared budget of $B$ channel
uses per slot, allocated by the BS in every slot (Section~\ref{sec:sharedsched}).

\emph{Uplink.} The mean signal-to-noise ratio (SNR) follows log-distance path loss,
$\bar\gamma_k=\bar\gamma_{\mathrm{ref}}(d_{\mathrm{ref}}/\max(d_k,d_{\min}))^{\alpha}$, and the
instantaneous SNR is $\bar\gamma_k|h_k|^2$ with Rayleigh block fading per slot; the
transmitter knows only $\bar\gamma_k$. A packet of $L$ bits at rate $R$ occupies
$s=\lceil (L+H)/R\rceil$ channel uses, where $H$ is a fixed per-packet overhead for the
header and the cyclic redundancy check (CRC), $H=0$ unless stated otherwise, and is
decoded iff
$\log_2(1+\bar\gamma_k|h_k|^2)\ge R$, so \cite{tse2005fundamentals}
\begin{equation}\label{eq:ps}
P_s(R;\bar\gamma_k)=\exp\!\big(-(2^R-1)/\bar\gamma_k\big).
\end{equation}
Packets of one agent in one slot share the fading block, so their outcomes are nested
in the rate. We also evaluate a finite-blocklength (FBL) model with the
normal-approximation error probability \cite{polyanskiy2010channel,durisi2016toward},
averaged over the fading, with an optional coding gap. Rates come from a finite set
$\mathcal{R}$. \emph{Conventional link adaptation} uses the largest $R\in\mathcal{R}$
with $1-P_s(R)\le0.1$; if none qualifies (below about 2.5\,dB mean SNR) it uses the
lowest rate, and packets larger than a slot are split into independently coded
per-slot fragments that are retransmitted until decoded. Hybrid automatic repeat request
(HARQ) acknowledgments are reliable.

\emph{Downlink.} The edge fuses decoded packets and broadcasts the result; by default
the broadcast is error-free, and Section~\ref{sec:robust} studies loss and delay.

\emph{Communication cost.} Our primary metric is the number of uplink channel uses,
including overheads. Section~\ref{sec:overhead} also reports the total cost, which comprises uplink data
and control (scheduling requests), the downlink broadcast of the common knowledge (its
payload plus $H$ per broadcast), one bit of HARQ feedback per uplink transmission, and
downlink grants, with downlink bits counted at 2\,bit per channel use.

\subsection{Shared and Private Knowledge}
Let $\ell^E$ be the edge log-odds map.
\begin{prop}[Knowledge decomposition]\label{prop:decomp}
With conditionally independent detections, exactly-once fusion at the edge and an
error-free, instantaneous broadcast, agent $k$'s belief is $\ell_k=\ell^E+\delta_k$,
where $\delta_k$ sums the LLRs of $k$'s own undelivered records.
\end{prop}
\begin{proof}
$\ell_k$ is the prior plus the LLRs of all records known to $k$, namely its own and, via the
broadcast, all delivered records of the others. $\ell^E$ is the prior plus all
delivered records; the difference is $k$'s undelivered records, and LLRs add.
\end{proof}
A sender therefore knows the shared component of its teammates' knowledge but not
their private components $\delta_j$; the broadcast is what closes the loop between the
team's knowledge and each sender's decisions. With a delayed or lossy broadcast, $\ell^E$
is replaced by the snapshot last received by agent $k$ (slot $v_k$), and $\delta_k$
comprises $k$'s own evidence \emph{not represented in that snapshot}, that is, records not yet
delivered plus records delivered after slot $v_k$. Our simulator implements exactly this
bookkeeping; delivered messages are never re-sent, because HARQ acknowledges them.

\subsection{Problem}
With $n_k^t$ the channel uses of agent $k$ in slot $t$ and a price $\eta\ge0$ per channel
use, we seek decentralized policies for
\begin{equation}\label{eq:problem}
\min\ \E\Big[T+\eta\sum_{t<T}\sum_{k}n_k^t\Big]\quad\text{s.t.}\quad n_k^t\le W\ \ \Big(\text{or}\ \textstyle\sum_k n_k^t\le B\Big),
\end{equation}
a decentralized partially observable Markov decision process (POMDP) that is
NEXP-complete in general \cite{bernstein2002complexity}.
In the rescue task, $T$ in \eqref{eq:problem} is replaced by the number of survivors not
saved. We do not solve \eqref{eq:problem}; each agent uses a myopic value-versus-cost
surrogate, whose structure we analyze in Section~\ref{sec:gosc} and whose merit we
evaluate empirically in Section~\ref{sec:results}.

\section{What, When and How Reliably in GOSC}\label{sec:gosc}
Figure~\ref{fig:arch} shows GOSC at each agent. The semantic representation fixes the
candidate messages (\emph{what}, Section~\ref{sec:enc}); a pluggable value function prices
them (Section~\ref{sec:val}); and one value-priced selection decides \emph{when} each is
sent and \emph{how reliably} (Section~\ref{sec:sel}). We then characterize this selection
(Sections~\ref{sec:sel}--\ref{sec:acc}) and extend it to a shared uplink
(Section~\ref{sec:sharedsched}).

\subsection{Semantic Representation for What to Send}\label{sec:enc}
GOSC sends \textsc{Confirm} (survivor cell, 18 bits), \textsc{Intent} (target cell, 18
bits) and \textsc{Evidence} messages. Evidence aggregates a run of $n$ pending records
losslessly using an anchor cell, 3-bit move codes (or a jump with $c_b+8$ bits), and, for
each detection, its record index and footprint offset, i.e.,
$L_E=34+3(n-1)+J(c_b+8)+P(\lceil\log_2 n\rceil+f_b)$ bits for $J$ jumps and $P$
detections, with $c_b=\lceil\log_2N^2\rceil$ and $f_b=\lceil\log_2(2r+1)^2\rceil$
(Table~\ref{tab:msgs}). A \emph{relevance filter} queues a record only if it would move
the belief held by the team after the queued evidence by at least $\varepsilon$ expected
survivors. The representation and the filter decide which information is a candidate at
all; which candidate is worth its radio cost is decided next.

\begin{table}[t]
\centering
\caption{Message types and payload sizes ($N=32$, $r=2$), without the per-packet
overhead $H$}
\label{tab:msgs}
\begin{tabular}{@{}llr@{}}
\toprule
Scheme & Content per message & Bits\\
\midrule
Raw & position + 8-bit score of each of 25 cells & 218\\
Natural language & English status report & $\approx$1{,}480\\
Semantic periodic & position, detections, intention & $33+5P$\\
\midrule
GOSC \textsc{Confirm} & survivor cell & 18\\
GOSC \textsc{Intent} & target cell & 18\\
GOSC \textsc{Evidence} & $n$ aggregated records, Sec.~\ref{sec:enc} & $34$ ($n{=}1$)\\
\bottomrule
\end{tabular}
\end{table}

\subsection{Message Values}\label{sec:val}
GOSC admits any value function; we consider four. (i) The \emph{ToM} value. The sender models each teammate $j$ as a Boltzmann-rational planner applying
\eqref{eq:score} to the shared knowledge $\mathcal{K}^E$,
$\pi_j(c)\propto\exp(S_j(c;\mathcal{K}^E)/(\vartheta\max S_j))$, and values a message $m$ by its
value of information (VoI), the change of the teammates' expected utility $U_j$ under its own, better-informed belief,
\begin{equation}\label{eq:voi}
\VoI_k(m)=\textstyle\sum_{j\neq k}\big(\E_{\pi_j(\cdot\mid\mathcal{K}^E\oplus m)}[U_j]-\E_{\pi_j(\cdot\mid\mathcal{K}^E)}[U_j]\big).
\end{equation}
In the rescue task, a claim receives a high value if some teammate would otherwise
choose the same survivor. (ii) \emph{Belief divergence} is the total-variation shift of
the shared belief. (iii) \emph{Throughput} is the payload size. (iv) \emph{Age of
information} is the summed age of the delivered records. Confirmations always have a
fixed high value. Only (i) models the receivers' decisions; (iii) and (iv) encode no task
knowledge beyond the relevance filter.

\subsection{Value-Priced Joint Selection for When and How Reliably}\label{sec:sel}
Each slot, the sender forms one \emph{group} $i$ per candidate message (each
confirmation, the intent if it changed, and the evidence). An option $o$ of group $i$
is a rate $R_{io}\in\mathcal{R}$ with the payload that fits into $W$ channel uses at that
rate (for evidence, the longest fitting prefix of the queue), occupying
$s_{io}=\lceil(L_{io}+H)/R_{io}\rceil$ channel uses, with expected value
$v_{io}=V_{io}P_s(R_{io};\bar\gamma)$, where $V_{io}$ is the value of its payload. The
per-slot problem is the multiple-choice knapsack \cite{kellerer2004knapsack}
\begin{equation}\label{eq:mck}
\max_x\textstyle\sum_{i,o}x_{io}\,u_{io}(\eta)\ \ \text{s.t.}\ \sum_{i,o}x_{io}s_{io}\le W,\ \sum_o x_{io}\le1,
\end{equation}
with $x_{io}\in\{0,1\}$ and net value $u_{io}(\lambda)=v_{io}-\lambda s_{io}$. For a price
$\lambda\ge\eta$, the \emph{decoupled rule} $x(\lambda)$ selects in each group the option
$o_i(\lambda)=\argmax_o u_{io}(\lambda)$ if this net value is positive and nothing
otherwise; $S(\lambda)$ denotes its channel uses. GOSC applies $x(\eta)$ if
$S(\eta)\le W$; otherwise it finds by bisection \cite{boyd2004convex} the smallest
$\lambda$ with $S(\lambda)\le W$, i.e., it relaxes the budget with a multiplier
$\mu=\lambda-\eta$, and fills the residual budget greedily (Algorithm~\ref{alg:gosc}). The
same rule answers \emph{when} (whether a group is selected) and \emph{how reliably} (which
rate), as the following result makes precise.

\begin{theorem}[Threshold structure for when and how reliably]\label{thm:structure}
Consider a message with a fixed payload $L$ and value $V\ge0$, a price $\lambda>0$ and
mean SNR $\bar\gamma$. Let $c(R;\bar\gamma)=\lceil(L+H)/R\rceil/P_s(R;\bar\gamma)$ be the
expected channel uses per successful delivery at rate $R$ (with independent
retransmissions), $c^\star(\bar\gamma)=\min_{R}c(R;\bar\gamma)$ over the rates whose packet
fits into the budget, and $R_c$ the largest minimizer. Then the following hold.
\begin{enumerate}[label=(\roman*),leftmargin=*,itemsep=0pt]
\item \emph{When.} The message is sent iff $V>\lambda c^\star(\bar\gamma)$, and
$c^\star(\bar\gamma)$ is non-increasing in $\bar\gamma$; hence neither a larger value nor a
better channel ever suppresses a transmission.
\item \emph{How reliably.} The selected rate $R^\star(V)$ (the largest maximizer of
$u(R)=VP_s(R;\bar\gamma)-\lambda\lceil(L+H)/R\rceil$) is non-increasing in $V$. For values
just above the threshold it minimizes $c$, and for sufficiently large values it is the
most robust feasible rate. Every transmitted message therefore uses a rate at or below
$R_c$; if $(L+H)/R$ is an integer for all $R$, the minimizers of $c$ maximize the goodput
$R\,P_s(R;\bar\gamma)$.
\end{enumerate}
\end{theorem}
\begin{proof}
(i) The message is sent iff $\max_R u(R)>0$, i.e., iff $V>\lambda s(R)/P_s(R)$ for some $R$.
By \eqref{eq:ps}, $P_s(R;\bar\gamma)$ increases in $\bar\gamma$, so each $c(R;\bar\gamma)$ and
their minimum are non-increasing in $\bar\gamma$. (ii) For $R_1<R_2$,
$u(R_2)-u(R_1)=V[P_s(R_2)-P_s(R_1)]+\lambda(s(R_1)-s(R_2))$ decreases in $V$ because $P_s$
decreases in $R$; $u$ has decreasing differences in $(R,V)$, and Topkis' theorem
\cite{topkis1998supermodularity} shows that the largest maximizer is non-increasing in
$V$. At $V=\lambda c^\star(1+\epsilon)$, $u(R)=\lambda P_s(R)\,[c^\star(1+\epsilon)-c(R)]$, which
is positive for $R=R_c$ and negative for every $R$ with $c(R)>c^\star(1+\epsilon)$; for
$\epsilon$ below the relative gap between the smallest and second-smallest values of $c$,
the maximizer is therefore a minimizer of $c$. As $V\to\infty$, $u(R)/V\to P_s(R)$, which
is maximized by the lowest feasible rate. Without ceilings, $c(R)=(L+H)/(R\,P_s(R))$.
\end{proof}
Theorem~\ref{thm:structure} answers \emph{when} and \emph{how reliably} with one number,
the price of a successful delivery $\lambda c^\star(\bar\gamma)$. A message is sent when its
value exceeds it, at the margin it is sent at the cost-efficient (goodput-maximizing)
rate that the rate-aware periodic baseline of Section~\ref{sec:results} uses for every
message, and the further its value exceeds the margin, the more of its channel uses are
spent on protection. Unequal error protection thus follows from the task rather than
being imposed. The ordering concerns one message; messages with different payloads are
compared through their thresholds. The argument only uses that $P_s$ increases in
$\bar\gamma$ and decreases in $R$, so it also holds for the fading-averaged FBL model; for
evidence, whose payload grows with the prefix, (i) holds with $V$ and $L$ indexed by
the option. In the logged instances of Section~\ref{sec:exact}, every intent and
confirmation selected by the decoupled rule used a rate at or below $R_c$, and
\ConfMinRate\% of the confirmations, whose value is fixed and high, used the most robust
feasible rate; the only exceptions (\IntentExcHi\% and \IntentExcLo\% of the intents at
$W=100$ and $W=50$) came from the greedy fill, where lower rates no longer fitted into
the residual budget.

\begin{algorithm}[t]
\caption{GOSC at agent $k$ in slot $t$}\label{alg:gosc}
\begin{algorithmic}[1]
\State Sense, update $\delta_k$; if $\sigma(\ell^E+\delta_k)(c)\ge\tau$, add $c$ to the pending confirmations
\State Select intention $g_k$ by \eqref{eq:score}
\State \emph{What.} Admit record $o_k^t$ to the queue if it passes the relevance filter
\State Build groups for confirmations (fixed value) and for intent and evidence prefixes (value, e.g., \eqref{eq:voi}), for every $R\in\mathcal{R}$
\State \emph{When and how reliably.} $\lambda\gets\eta$; \textbf{if} $S(\eta)>W$ \textbf{then} bisect $\lambda$ until $S(\lambda)\le W$
\State Transmit $x(\lambda)$; greedily add the remaining groups that fit
\State On HARQ acknowledgment, remove delivered records from the queue; on broadcast, update $\delta_k$ (closes the loop)
\end{algorithmic}
\end{algorithm}

\subsection{Exactness of the Lagrangian Scheduler}\label{sec:exact}
\begin{theorem}[Exactness and optimality gap]\label{thm:exact}
Let $\mathrm{OPT}$ be the optimum of \eqref{eq:mck} and $J(x)=\sum_{i,o}x_{io}u_{io}(\eta)$.
\begin{enumerate}[label=(\roman*),leftmargin=*,itemsep=0pt]
\item If $S(\eta)\le W$, then $x(\eta)$ is optimal.
\item For every $\lambda\ge\eta$ with $S(\lambda)\le W$,
$J(x(\lambda))\le\mathrm{OPT}\le J(x(\lambda))+(\lambda-\eta)\,(W-S(\lambda))$. In particular,
$x(\lambda)$ is optimal if $S(\lambda)=W$, and it is always optimal for the budget
$S(\lambda)$.
\item Let $\lambda^\star=\inf\{\lambda\ge\eta:S(\lambda)\le W\}>\eta$ and suppose that one group
$i^\star$ changes its option at $\lambda^\star$, from $o^-$ (below $\lambda^\star$) to $o^+$
(above), as is generic. Then
$\mathrm{OPT}-J(x(\lambda^\star{+}))<u_{i^\star o^-}(\eta)-u_{i^\star o^+}(\eta)\le\max_{i,o}u_{io}(\eta)$;
hence the relaxation loses less than the net value of one message.
\end{enumerate}
\end{theorem}
\begin{proof}
For any feasible $x$ and $\lambda\ge\eta$,
$J(x)=\sum x_{io}u_{io}(\lambda)+(\lambda-\eta)\sum x_{io}s_{io}\le\sum_i\max_o u_{io}(\lambda)^+ +(\lambda-\eta)W
=J(x(\lambda))+(\lambda-\eta)(W-S(\lambda))$, because $x(\lambda)$ maximizes
$\sum x_{io}u_{io}(\lambda)$ under the group constraints alone. This is (ii); replacing $W$
by $S(\lambda)$ gives optimality for that budget \cite{everett1963generalized}, and (i) is
the case $\lambda=\eta$. For (iii), $x^-=x(\lambda^\star{-})$ and $x^+=x(\lambda^\star{+})$ both
maximize the Lagrangian at $\lambda^\star$, so
$J(x^-)-J(x^+)=(\lambda^\star-\eta)(S(x^-)-S(x^+))$; with $S(x^-)>W$ and (ii),
$\mathrm{OPT}-J(x^+)\le(\lambda^\star-\eta)(W-S(x^+))<J(x^-)-J(x^+)$, which is the change of
the net value of group $i^\star$, and $u_{i^\star o^+}(\eta)\ge0$.
\end{proof}
Theorem~\ref{thm:exact} explains why the heuristic is almost always exact. On 12{,}246 and
13{,}125 logged scheduling instances ($W=100$ and $W=50$; 30 search missions each), the
price $\eta$ alone kept the budget slack in \SlackHi\% and \SlackLo\% of the instances,
where (i) certifies optimality; in the remaining instances the relaxation was still
exact in \BindExactHi\% and \BindExactLo\% of the cases, with mean relative gaps of
\BindGapHi\% and \BindGapLo\% there, and 0.22\% and 0.14\% over all instances. Event
triggering and the knapsack thus interact favourably; a scheduler that sends only what is
worth its price rarely saturates its budget, and whenever it does not, the decoupled rule
is optimal. These statements concern the per-slot surrogate \eqref{eq:mck}, not mission
optimality.

\subsection{How Accurate Must the Value Be?}\label{sec:acc}
Values such as \eqref{eq:voi} rest on a model of the teammates and are necessarily
approximate. Let the scheduler use estimates $\hat V_{io}=\xi_iV_{io}$, where the error
factor $\xi_i\in[e^{-\delta},e^{\delta}]$ of message $i$ is common to its options, and let the
\emph{threshold ratio} of message $i$ be $\theta_i=\min_o \lambda s_{io}/v_{io}$ (the message is
sent under the true values iff $\theta_i<1$).
\begin{prop}[Robustness to value errors]\label{prop:robust}
For a fixed price $\lambda\ge\eta$, let $o^\star_i$ and $\hat o_i$ be the decoupled choices
under the true and the estimated values (``no transmission'' being an option with
$s=v=0$). Then
\begin{enumerate}[label=(\alph*),leftmargin=*,itemsep=0pt]
\item the decision to send message $i$ can differ only if
$e^{-\delta}<\theta_i<e^{\delta}$;
\item $0\le u_{io^\star_i}(\lambda)-u_{i\hat o_i}(\lambda)\le(e^{\delta}-1)\,\lambda\,
|s_{io^\star_i}-s_{i\hat o_i}|\le(e^{\delta}-1)\lambda W$.
\end{enumerate}
\end{prop}
\begin{proof}
(a) Message $i$ is sent under the estimates iff $\xi_iv_{io}>\lambda s_{io}$ for some $o$,
i.e., iff $\xi_i>\theta_i$. (b) The estimated choice maximizes
$\xi_i(v_{io}-\lambda' s_{io})$ with $\lambda'=\lambda/\xi_i$, i.e., it is the true choice at
price $\lambda'$. The function $g(\ell)=\max_o(v_{io}-\ell s_{io})$ is convex with
subgradient $-s_{io(\ell)}$, so $g(\lambda)-g(\lambda')\le(\lambda'-\lambda)s_{io^\star}$, and the loss
$g(\lambda)-[g(\lambda')-(\lambda-\lambda')s_{i\hat o}]\le(\lambda'-\lambda)(s_{io^\star}-s_{i\hat o})$;
finally $|\lambda'-\lambda|=\lambda|1/\xi_i-1|\le\lambda(e^{\delta}-1)$.
\end{proof}
Value errors are therefore cheap in two respects. They affect only \emph{marginal}
messages, whose value is within a factor $e^{\delta}$ of the price of their delivery, and
the loss on such a message is bounded by the price of the channel uses it gains or loses,
independently of how large its value is. Messages that are clearly valuable
(confirmations, contested claims) or clearly irrelevant keep their decisions. In the
logged instances, \ZeroValue\% of the candidate intent and evidence messages had zero
ToM value, which no multiplicative error changes, and only \WithinOne\% had a threshold
ratio within a factor $e$ of one (\WithinTwo\% within $e^{2}$); every other decision is
unchanged by any error up to that factor. Proposition~\ref{prop:robust} concerns one slot
at a fixed price; Section~\ref{sec:noise} measures the effect of value errors on whole
missions, including errors that are not bounded by any $\delta$.

\subsection{Price Coordination on a Shared Uplink}\label{sec:sharedsched}
When all agents share $B$ channel uses per slot, the knapsack \eqref{eq:mck} couples all
agents' groups through one budget, and its multiplier becomes a network \emph{shadow
price} \cite{kelly1998rate}. GOSC therefore needs no central solver. The BS broadcasts a
price $\mu_t$ (8 bits per slot); each agent applies the decoupled rule at
$\lambda=\eta+\mu_t$ and requests the channel uses of its selection together with its
value density (a 16-bit request); the BS grants the requests in decreasing value density
while they fit and updates $\mu_{t+1}=\mu_t\exp\!\big(\kappa(D_t/B-1)\big)$, where $D_t$ is
the requested total ($\kappa=1$; $\mu$ is reset to zero below $10^{-7}$ and restarted from
$10^{-7}$ when demand exceeds $B$). Each agent's rule is the decoupled rule at a common
price, so Theorem~\ref{thm:structure} continues to hold with $\lambda=\eta+\mu_t$. The
periodic baselines send an 8-bit buffer status report \cite{3gpp38321}, and the BS divides
$B$ among them by max-min fair water-filling; every grant costs 32 downlink bits. A
central alternative in which each agent reports all its Lagrangian-relevant options, so
that the BS solves the network-wide knapsack exactly, used more control signalling
than it saved in our tests, and we do not pursue it.

\subsubsection*{Complexity}
Evaluating \eqref{eq:voi} costs one box filter (via integral images) and one softmax per
teammate over the $N^2$ cells, so building the groups costs
$\mathcal{O}(|\mathcal{R}|KN^2)$ per agent and slot, and the bisection
$\mathcal{O}(I\,G\,|\mathcal{R}|)$ for $I$ iterations and $G$ groups. With $N=32$, $K=6$
and $|\mathcal{R}|=9$, our single-threaded Python implementation needs 2.0\,ms per
decision (median; 95th percentile 6.8\,ms) on an Apple M1 CPU; the throughput and AoI
values are cheaper still.

\section{Performance Evaluation}\label{sec:results}
We first establish the main result, resource efficiency against tuned semantic periodic
schedules (Section~\ref{sec:eff}), and then examine each of the three decisions, namely
\emph{what} (Section~\ref{sec:what}), \emph{when}, including how accurate the value must
be (Section~\ref{sec:noise}), and \emph{how reliably} (Section~\ref{sec:how}). We then
test the result under more realistic radio resources (Section~\ref{sec:realism}) and
with LLM agents (Section~\ref{sec:llm}).

\begin{table}[t]
\centering
\caption{Default parameters}
\label{tab:params}
\resizebox{\columnwidth}{!}{%
\begin{tabular}{@{}ll@{}}
\toprule
Parameter & Value\\
\midrule
Grid, cell size, slot & $32\times32$, 10\,m, 1\,s\\
Agents $K$, survivors $V$ & 6, 8 (uniform; clustered in Sec.~\ref{sec:eff})\\
Footprint radius $r$; $p_d$, $p_{fa}$; $\tau$ & 2 ($5\times5$); 0.85, 0.05; 0.99\\
Planner $\beta$, $\varsigma$, $h$ & 0.1, 5 cells, 1.2\\
Channel uses per agent and slot $W$ & 100\\
$\bar\gamma_{\mathrm{ref}}$ at 100\,m; path-loss exponent $\alpha$ & 5\,dB; 3\\
Rate set $\mathcal{R}$ (bit/channel use) & $\{0.25,0.5,1,1.5,2,3,4,5,6\}$\\
Conventional link adaptation; data lifetime & block error rate $\le$ 10\%; 20 slots\\
Per-packet overhead $H$ & 0 (16, 32, 64 in Sec.~\ref{sec:overhead})\\
Rescue deadlines; service time & $\mathcal{U}[60,180]$; 8 slots\\
GOSC $\eta$, $\varepsilon$, $\vartheta$, $v_C$ & $10^{-5}$, 0.02, 0.1, 5\\
\bottomrule
\end{tabular}}
\end{table}

\begin{table*}[t]
\centering
\caption{Main result. Fresh-seed validation (seeds 3000--3099) of frozen operating
points, with the target and mean outcome (completion time in slots, or survivors saved in \%),
the margin to the target (positive when met with room) and uplink channel uses ($\times10^3$),
for GOSC and for the most frugal of the 38 tuned periodic configurations meeting the
target on the evaluation seeds; paired outcome difference (GOSC minus periodic) and
channel-use ratio (periodic over GOSC). Brackets show 95\% bootstrap CIs over the
fresh seeds}
\label{tab:val}
\resizebox{\textwidth}{!}{%
\begin{tabular}{@{}lrrlrrlrll@{}}
\toprule
 & & \multicolumn{3}{c}{GOSC} & \multicolumn{3}{c}{Best tuned periodic} & & \\
Setting & Target & Outcome & Margin & Ch.\ uses & Outcome & Margin & Ch.\ uses & Outcome diff. & \textbf{Ratio}\\
\midrule
\csname @@input\endcsname figures/table_validation.tex
\bottomrule
\end{tabular}}
\end{table*}

\subsection{Setup, Baselines and Statistics}
Table~\ref{tab:params} lists the default parameters ($T_{\max}=400$ slots). GOSC's
hyper-parameters ($\eta=10^{-5}$, $\varepsilon=0.02$, $\vartheta=0.1$) and the rescue
deadlines were fixed on calibration seeds 1000--1079; all results use the disjoint
evaluation seeds 0--99, with common random numbers across schemes.

\emph{Baselines.} All schemes share the agents, the planner and the edge fusion, and
serve confirmations with strict priority. The strongest baselines form a family of
\emph{tuned semantic periodic schedules}. \emph{Rate-aware periodic} sends every $k$
slots one packet with the last $k$ records in GOSC's aggregated encoding plus the
intention, at the rate that maximizes the expected goodput $R\,P_s(R)$, with
fragmentation; we sweep $k\in\{1,2,4,8,16\}$, with and without GOSC's relevance filter.
\emph{Information-suppressing periodic} retains only records whose belief change exceeds
a threshold $\theta\in\{0.02,0.05,0.1,0.2,0.4,1,2\}$ (expected survivors, the same
measure as GOSC's filter), sends a short intent-only packet when only the intention
changed, and otherwise stays silent, for $k\in\{1,2,4,8\}$. We always use the best of
these 38 configurations. Conventional baselines send every slot with conventional link
adaptation. \emph{Semantic periodic} sends a structured update, \emph{raw} sends the position
and an 8-bit score per footprint cell, and \emph{natural language} sends an English status
report;
\emph{report only} sends confirmations only. The \emph{genie} delivers every update
instantly and error-free; it is an \emph{ideal-communication reference}, not a
performance bound, because more information can also change the heuristic planner's
behaviour unfavourably.

\emph{Statistics.} We compare schemes by paired differences over seeds with 95\%
percentile-bootstrap confidence intervals (CIs; 4{,}000 resamples), written $\Delta$ [lower,
upper]; an interval containing zero means that no significant difference was detected,
not equivalence. Our main comparisons are at \emph{matched outcomes}; targets are set
relative to the completion time $T_{\mathrm g}$ (search) or the saved fraction
$S_{\mathrm g}$ (rescue) of the ideal-communication reference, and each method's operating
point with the fewest uplink channel uses whose mean outcome meets the target is
selected on the evaluation seeds, \emph{frozen}, and re-run on fresh seeds 3000--3099
(frozen points in Appendix~\ref{app}, Table~\ref{tab:frozen}). This avoids
interpolation and selection bias.


\subsection{Main Result on Resource Efficiency Against Tuned Semantic Schedules}\label{sec:eff}
Table~\ref{tab:val} is our main result. In all six settings, GOSC reached the common
target with \textbf{1.2--8.5 times fewer uplink channel uses} than the most frugal
configuration of the tuned periodic family on the fresh seeds, and every ratio interval
lies above one. The advantage grows as resources become scarce, reaching 1.8, 2.6 and 8.5 times in
the rescue task at $W=100$, $W=50$ and $-5$\,dB, and 5.9 times with $K=10$ agents, whose
periodic updates compete for the same knowledge. The savings are not bought with
outcomes. In the rescue task, GOSC met the target with a significant margin in all
three settings and saved significantly more survivors than the periodic configuration
at $W=100$ (+3.6 [1.2, 6.1] percentage points) and $W=50$ (+7.2 [4.6, 10.0]); at
$-5$\,dB, no outcome difference was detected. In the search task, both methods meet or
narrowly miss the targets on the fresh seeds, and neither the margins nor the outcome
differences are significant. The selected periodic configurations are those with
filtering or suppression, so the savings do not come merely from allowing GOSC to
withhold updates; they come from deciding \emph{which} update to send \emph{when}, and
how robustly.

Figure~\ref{fig:rescue} shows the same picture across operating points;
the GOSC frontiers extend to fewer channel uses than any periodic configuration and
dominate there, and the curves converge as the channel use grows. Where resources are
abundant, periodic reporting is as good or better (Section~\ref{sec:realism}); GOSC's
contribution is efficiency under constrained resources.

\begin{figure}[t]
\centering
\includegraphics[width=\columnwidth]{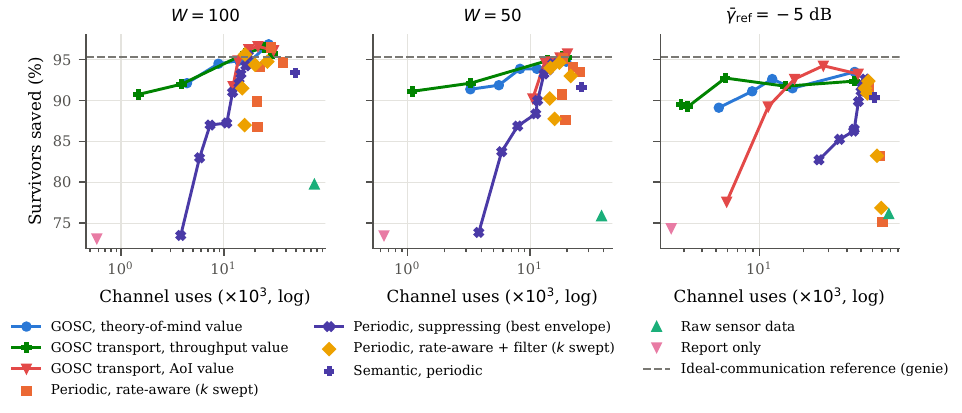}
\caption{Rescue task with deadlines, showing survivors saved versus uplink channel uses for
$W=100$, $W=50$ and $\bar\gamma_{\mathrm{ref}}=-5$\,dB (100 missions per point). Value-priced
rules sweep $\eta$; squares and diamonds are rate-aware periodic configurations without
and with the relevance filter ($k$ swept); the violet curve is the best envelope of the
information-suppressing periodic variant; single markers are conventional baselines;
the dashed line is the ideal-communication reference.}
\label{fig:rescue}
\end{figure}

\begin{table*}[t]
\centering
\caption{What to send, comparing representations in the search task at the default
operating point (100 missions; $\pm$ gives the 95\% CI of the mean; $\Delta$ is the paired
difference in completion time to GOSC with 95\% bootstrap CI; false declarations are
survivors declared at a wrong cell, per mission)}
\label{tab:default}
\setlength{\tabcolsep}{4pt}
\begin{tabular}{@{}lrrrrrr@{}}
\toprule
Scheme & Completion & $\Delta$ vs.\ GOSC & Payload & Channel uses & Confirm.\ latency & False\\
 & time (slots) & (slots) & (kbit) & ($\times10^3$) & (slots) & declarations\\
\midrule
\csname @@input\endcsname figures/table_default.tex
\bottomrule
\end{tabular}
\end{table*}

\subsection{Semantic Representation for What to Send}\label{sec:what}
Table~\ref{tab:default} compares representations in the search task at the default
operating point. GOSC completes the mission faster than raw-data sharing ($\Delta=60.9$
[48.1, 74.2] slots), natural-language messaging (77.7 [67.7, 88.3]) and report-only
operation (115.7 [98.5, 133.7]). Against semantic periodic updates, no significant
difference in completion time was detected (5.8 [$-2.4$, 13.8]), while GOSC used 3.4
times fewer channel uses. The natural-language comparison is illustrative rather than
a central claim; an English report of about 1{,}480 bits against a
typed message of a few tens of bits is bound to lose, and its reports occupy at least
eight slots each, so that even confirmations wait 32.8 slots on average. The informative
comparison is the one against tuned semantic schedules in Section~\ref{sec:eff}. Every
scheme completed all 100 missions within $T_{\max}$, and false declarations stayed below
0.2 per mission.

\subsection{When to Send, and How Accurate Must the Value Be?}\label{sec:noise}
GOSC's trigger compares a message's value with the price of its delivery
(Theorem~\ref{thm:structure}). We therefore ask \emph{how accurate this value needs to be for
resource-efficient communication}. Three experiments answer it.

\subsubsection*{Value functions at matched resources}
Table~\ref{tab:values} compares the four value functions on
GOSC's transport at matched channel uses of 5k, 10k and 20k. The decision-aware ToM value
is significantly better than belief divergence, by 6.9--16.5 slots in seven of ten
comparisons, but against the throughput or AoI values, which encode no model of the
teammates, it is significantly better in only one of fifteen comparisons ($K=10$ at
20k, 5.0 [0.9, 7.6] slots), which does not survive a Bonferroni correction, and
significantly worse in none. In the rescue task, the throughput value is as efficient as
the ToM value or more (Figure~\ref{fig:rescue}).

\begin{table}[t]
\centering
\caption{Valuation rules at matched radio budgets, showing the completion time of GOSC
(slots) and the paired difference of each rule to GOSC, $\Delta$ [95\% bootstrap CI]
(positive means GOSC is faster; -- means the budget is outside the rule's frontier). The Bonferroni
correction covers the fifteen comparisons with the simpler values}
\label{tab:values}
\setlength{\tabcolsep}{3pt}
\resizebox{\columnwidth}{!}{%
\begin{tabular}{@{}lrrlll@{}}
\toprule
Setting & Budget & GOSC & Belief div. & Throughput & AoI\\
\midrule
\csname @@input\endcsname figures/table_values.tex
\bottomrule
\end{tabular}}
\end{table}

\subsubsection*{Imperfect values}
We then degraded GOSC's ToM value deliberately (Appendix~\ref{app}), using a log-normal error
$e^{\sigma z}$ per message with $\sigma=0.5$, 1 and 2 (typical errors of a factor 1.6, 2.7
and 7.4), values known only to their order of magnitude, and \emph{content-independent}
values that keep the distribution of the true values but none of their information,
with and without the relevance filter. Their Spearman correlations with the exact value
are $\rho=0.99$, 0.95, 0.89, 0.95 and 0.14. For every estimate and setting, we re-selected
the price with the rule of Section~\ref{sec:eff}, froze it and re-ran it on the fresh
seeds against the frozen periodic configuration of Table~\ref{tab:val}; in the rescue
task, the price grid was extended to $10^{-3}$ for all estimates, including the exact
one. Table~\ref{tab:noise} and Figure~\ref{fig:noise} show three results.

\emph{Precise values are not necessary.} With the four estimates whose rank correlation
is at least 0.89, GOSC met the target with significantly fewer channel uses than the
tuned periodic schedules in 20 of 24 combinations of setting and estimate; even errors
of a typical factor of seven ($\sigma=2$) left 1.3--9.6-fold reductions. Of the four
exceptions, two show no significant difference (clustered search with $\sigma=1$ and with
orders of magnitude), one misses the target (orders of magnitude with $K=10$), and one
reverses, namely $\sigma=1$ in the default search setting, where both families meet the target
only narrowly and the frozen price is fragile ($\sigma=0.5$ and $\sigma=2$ gave 1.75 and
1.27 there).

\emph{Coarse relevance is necessary.} Content-independent values still gave 2.0--7.2-fold
reductions in the rescue task, where the fixed priority of confirmations and the
price-based aggregation carry most of the benefit, but they lost the advantage in the
default search task (0.57), and without the relevance filter also with $K=10$ agents
(0.81; 1.33 with the filter).

\emph{Degradation is graceful.} The reduction shrinks smoothly as the accuracy decreases,
e.g., from 5.9 to 5.2, 3.7 and 1.8 in the rescue task. With the extended price grid, the
exact value itself reduced channel uses 5.9, 4.4 and 13.0-fold in the three rescue
settings, so the main result of Table~\ref{tab:val} is conservative (at $-5$\,dB the
exact value narrowly missed the target on the fresh seeds, by 0.6 [$-1.8$, 3.2]
percentage points, which is not significant).

\begin{table*}[t]
\centering
\caption{How accurate must the value be? Channel-use ratio of the frozen tuned periodic
configuration (Table~\ref{tab:val}) to GOSC with degraded value estimates, each with its
own frozen price, on the fresh seeds [95\% bootstrap CI]; $\rho$ is the Spearman correlation of
the estimate with the exact ToM value, and $\dagger$ marks that GOSC's mean outcome misses
the target on the fresh seeds. Rescue prices extended to $10^{-3}$ for every estimate}
\label{tab:noise}
\setlength{\tabcolsep}{3pt}
\resizebox{\textwidth}{!}{%
\begin{tabular}{@{}lrllllll@{}}
\toprule
Value estimate & $\rho$ & Search & Search, $K{=}10$ & Search, clustered & Rescue & Rescue, $W{=}50$ & Rescue, $-5$\,dB\\
\midrule
\csname @@input\endcsname figures/table_noise_r2.tex
\bottomrule
\end{tabular}}
\end{table*}

\begin{figure}[t]
\centering
\includegraphics[width=.8\columnwidth]{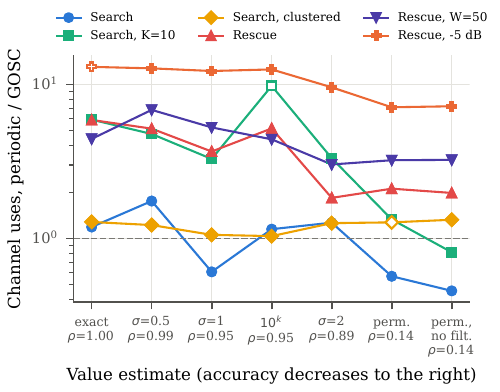}
\caption{How accurate must the value be? Ratio of the uplink channel uses of the frozen
tuned periodic configuration to those of GOSC with degraded value estimates (fresh
seeds, each estimate with its own frozen price; log scale); above the dashed line GOSC
is cheaper. Hollow markers indicate that GOSC's mean outcome misses the target. Estimates are ordered
by accuracy; $\rho$ is the rank correlation with the exact value.}
\label{fig:noise}
\end{figure}

\subsubsection*{Why coarse values suffice}
Proposition~\ref{prop:robust} explains the first result; a value error changes only
decisions on messages near their send threshold, and only \WithinOne\% of the candidate
messages were within a factor $e$ of it. A counterfactual diagnostic explains why even
exact values add little; a single message changed
a teammate's next target in only 6--12\% of 914 sampled decisions, its effect on the
completion time had both signs, and none of the tested scores, including the ToM value,
ranked high-impact messages significantly higher. Exact valuation of message \emph{sets}
did not help either; an exact set-valued scheduler that values the received set with the
joint ToM value and solves each slot by enumeration did not improve the completion time
in any of 1{,}800 paired missions, and was significantly slower in one setting (9.8 [1.9,
18.3] slots). Our answer to the research question is therefore that \emph{semantic valuation
for resource-efficient agent communication must be coarsely right, not precise}; it must
rank messages roughly correctly and encode task relevance, while the precision that
decision-aware models such as theory of mind aim for is not needed once the scheduler is
communication-aware.

\subsection{Value-Aware Rate Selection for How Reliably to Deliver}\label{sec:how}
Table~\ref{tab:ablation} removes one component at a time at the default price. Restricting
GOSC to the conventional link-adaptation rate, with a fallback to the lowest higher rate
at which a packet fits into a slot, costs 15.8 [7.7, 24.2] slots at $W=100$, while using
8.0k instead of 12.1k channel uses; at $W=50$, the difference (10.0 [$-0.8$, 20.2]) is not
significant. The rate is where GOSC spends channel uses on protection when a message is
worth it (Theorem~\ref{thm:structure}); the confirmations, whose value is highest, used the
most robust feasible rate in \ConfMinRate\% of the cases. The relevance filter has no
significant effect at the default price, and replacing the ToM value by the throughput
value changes the completion time insignificantly while, at the same price, using two to
three times more channel uses (at matched channel uses, Table~\ref{tab:values}). The
ablations change resource use as well as outcome, so they compare configurations rather
than equal-cost alternatives.

\begin{figure}[t]
\centering
\includegraphics[width=\columnwidth]{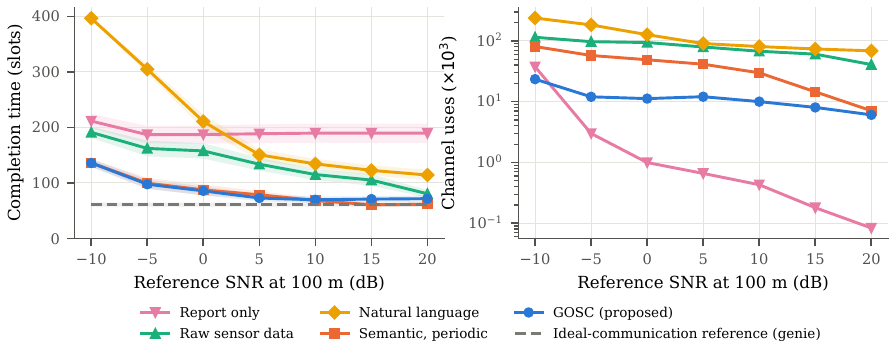}
\caption{Mission completion time (left) and uplink channel uses (right, log scale)
versus the reference SNR ($W=100$, $K=6$; outage model; 100 missions per point). Shaded
bands show 95\% confidence intervals of the mean. Completion times are capped at
$T_{\max}=400$ slots; the ideal-communication reference does not use the channel.}
\label{fig:snr}
\end{figure}

\subsubsection*{Channel quality and short packets}
Figure~\ref{fig:snr} sweeps the reference SNR. Against semantic periodic updates, no
significant difference in completion time was detected from $-10$ to 10\,dB, while GOSC
used 3.0--4.8 times fewer channel uses; at 15 and 20\,dB, the periodic scheme was faster
by about 10 slots, because with a good link sending everything is cheap. Natural
language degrades most as the SNR drops and finished only 13\% of the missions at
$-10$\,dB. Under the FBL model, the fading-averaged success
probability differs from \eqref{eq:ps} by at most 0.086 for the packets used (median
0.001), the ranking of the schemes is unchanged with and without a 2\,dB coding gap, and
GOSC used 1.2--4.9 times fewer channel uses than semantic periodic updates over the whole
sweep.

\subsection{Realistic Radio Resources}\label{sec:realism}
Two simplifications could favour short semantic messages. So far, every packet was free of
overhead, and every agent owned its channel uses. We now remove both, and add the
downlink to the cost.

\subsubsection{Per-packet overhead and total uplink--downlink cost}\label{sec:overhead}
GOSC's messages carry only 18--34 bits of payload, which headers, CRC and HARQ signalling
could dominate. We therefore added a fixed overhead of $H=16$, 32 and 64 bits to every
uplink packet, roughly a 16-bit transport-block CRC \cite{3gpp38212}, the CRC plus a
medium access control (MAC) subheader \cite{3gpp38321}, and a full layer-2 header stack. GOSC pays it on every
message, whereas the periodic schemes pay it once per slot for a transport block that
multiplexes their queue; the downlink broadcast carries it too, and HARQ feedback costs
one downlink bit per uplink transmission. For every $H$, both families were re-selected
on the evaluation seeds and re-run on the fresh seeds (Table~\ref{tab:overhead}).

\begin{table}[t]
\centering
\caption{How reliably, via component ablations at the default price in the search task
(100 missions; $\Delta$ is the paired difference to GOSC in slots). The variants differ in channel use
as well as outcome}
\label{tab:ablation}
{\footnotesize
\begin{tabular}{@{}r>{\raggedright\arraybackslash}p{2.15cm}rrl@{}}
\toprule
$W$ & Variant & Completion & Ch.\ uses & $\Delta$ [95\% CI]\\
 & & (slots) & ($\times10^3$) & \\
\midrule
\csname @@input\endcsname figures/table_ablation.tex
\bottomrule
\end{tabular}}
\end{table}

\emph{The advantage survives the overhead.} In five of the six settings, GOSC used
significantly fewer uplink channel uses than the re-tuned periodic family at every
overhead (in four of these 20 cases, one of the two families narrowly missed the target
on the fresh seeds, never significantly), and in the rescue task the advantage \emph{grows} with the overhead, from 1.8
to 2.5, 3.0 and 4.1 times at $W=100$, from 2.6 to 6.0 times at $W=50$, and from 8.5 to
16.6 times at $-5$\,dB. Theorem~\ref{thm:structure} explains why. The overhead enters the
cost of a delivery, so GOSC's send threshold rises and it sends fewer, fuller packets
(at $-5$\,dB, 95, 89, 82 and 54 transmissions per mission as $H$ grows), whereas the
periodic schedules keep their cadence and pay the overhead on every fragment (527 to 636
transmissions). The exception is the default search setting, where both families only
narrowly meet the target; the frozen ratio was 0.84 at $H=16$ (the interpolated frontier
ratio, 1.17 [0.82, 2.95], shows no significant difference), 1.62 at $H=32$ with GOSC
narrowly and insignificantly missing the target on the fresh seeds, and at $H=64$ only
GOSC reached the target at all.

\emph{Total cost.} The last block of Table~\ref{tab:overhead} adds the downlink broadcast,
HARQ feedback and the uplink control to the uplink data. Because the broadcast relays
what the uplink delivered, it is itself proportional to the useful uplink traffic, and
the ratios of the total cost are only slightly below the uplink ratios (e.g., 1.6 instead
of 1.8 in the rescue task). Table~\ref{tab:cost} breaks the total down for the default
search setting; the broadcast of the common knowledge adds 3.3k channel uses to GOSC's
11.8k (22\% of the total) and 3.7k to the periodic family's 14.0k, and HARQ feedback is
negligible. Raw data and natural language cost about six times more than GOSC in total.

\begin{table}[t]
\centering
\caption{Per-packet overhead $H$ (bits) and the fresh-seed channel-use ratio of the re-tuned
periodic family to GOSC [95\% bootstrap CI], for the uplink and for the total uplink and
downlink cost (broadcast with its overhead, HARQ feedback). $\dagger$ and $\ddagger$ mark that GOSC's or
the periodic configuration's mean outcome misses the target on the fresh seeds; ``only
GOSC'' means that no periodic configuration met the target on the evaluation seeds}
\label{tab:overhead}
\setlength{\tabcolsep}{2.5pt}
\resizebox{\columnwidth}{!}{%
\begin{tabular}{@{}lllll@{}}
\toprule
Setting & $H=0$ & $H=16$ & $H=32$ & $H=64$\\
\midrule
\csname @@input\endcsname figures/table_overhead_r2.tex
\bottomrule
\end{tabular}}
\end{table}

\begin{table}[t]
\centering
\caption{Total communication cost in the default search setting (fresh seeds; frozen
operating points re-selected for each $H$), listing the completion time (slots), uplink channel uses
including overhead, overhead bits, downlink broadcast of the common knowledge, downlink
HARQ feedback, and the total (all $\times10^3$ channel uses unless noted). $\dagger$ marks a mean
outcome that misses the target ($T\le76.2$ slots); UL and DL denote uplink and downlink}
\label{tab:cost}
\setlength{\tabcolsep}{3pt}
\resizebox{\columnwidth}{!}{%
\begin{tabular}{@{}lrrrrrr@{}}
\toprule
Scheme & $T$ & UL & Overh. & DL bcast & DL fb. & \textbf{Total}\\
 & & & (kbit) & & & \\
\midrule
\csname @@input\endcsname figures/table_cost_r2.tex
\bottomrule
\end{tabular}}
\end{table}

\subsubsection{Shared network-wide uplink}\label{sec:shared}
All agents now compete for a shared budget of $B$ channel uses per slot, allocated by the
BS in every slot (Section~\ref{sec:sharedsched}), for GOSC by its broadcast network price and
for the 38 periodic configurations by max-min fair grants on buffer status reports. Uplink cost
now includes the scheduling requests. Table~\ref{tab:shared} repeats the fresh-seed
validation, and Figure~\ref{fig:shared} shows the search frontiers. The competition does not
change the picture. In the rescue task ($B=300$), GOSC met the target with 4.8 [4.4, 5.2]
times fewer uplink channel uses (4.6 [4.3, 5.0] times in total uplink and downlink cost).
In the search task, the results mirror those with orthogonal budgets. With $B=600$, the
same total as $K$ agents with $W=100$ each, GOSC needed 1.16 [1.10, 1.24] times fewer
channel uses (1.19 with orthogonal budgets); with $B=300$, as with $W=50$, no significant
difference was detected (0.95 [0.89, 1.02]); and with $B=150$ neither family reached the
target, but GOSC reached the best completion time of the periodic family (82.4 versus
82.6 slots) with 9.4k instead of 14.5k channel uses (Figure~\ref{fig:shared}). The shared
price lets GOSC coordinate the agents without a central solver; an agent whose messages
are worth less than the current network price defers them to a later slot.

\begin{table*}[t]
\centering
\caption{Shared network-wide uplink budget $B$ (all agents compete; $K=6$), with fresh-seed
validation of frozen operating points as in Table~\ref{tab:val}. Uplink channel uses
include scheduling requests; the last column adds the downlink broadcast, HARQ feedback,
grants and the price broadcast. UL and DL denote uplink and downlink. Brackets show 95\% bootstrap CIs over the fresh seeds}
\label{tab:shared}
\resizebox{\textwidth}{!}{%
\begin{tabular}{@{}lrrlrrlrlll@{}}
\toprule
 & & \multicolumn{3}{c}{GOSC (price-coordinated)} & \multicolumn{3}{c}{Best tuned periodic (fair grants)} & & \multicolumn{2}{c}{Ratio periodic/GOSC}\\
Setting & Target & Outcome & Margin & UL ch.\ uses & Outcome & Margin & UL ch.\ uses & Outcome diff. & Uplink & UL + DL\\
\midrule
\csname @@input\endcsname figures/table_shared_r2.tex
\bottomrule
\end{tabular}}
\end{table*}

\begin{figure}[t]
\centering
\includegraphics[width=\columnwidth]{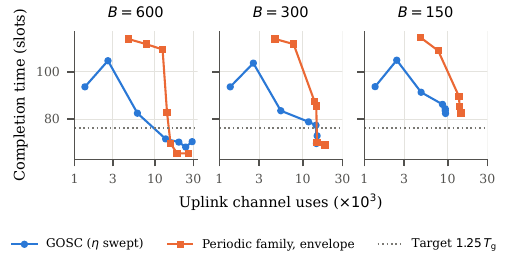}
\caption{Search task with a shared network-wide uplink budget $B$, showing completion time versus
uplink channel uses (including scheduling requests) per mission on the evaluation seeds
(100 missions per point). GOSC sweeps $\eta$; the orange line is the best envelope of the
38 periodic configurations with max-min fair grants; the dotted line is the target $1.25\,T_{\mathrm g}$.}
\label{fig:shared}
\end{figure}

\subsubsection{Downlink impairments}\label{sec:robust}
We also impaired the broadcast of the common knowledge. With a lossy broadcast, GOSC is faster than semantic periodic
updates at loss probabilities 0.3 and 0.5 (13.1 [4.9, 21.8] and 10.2 [2.0, 18.2] slots);
for delays of up to 10 slots, no significant difference was detected, and all schemes,
including the ideal-communication reference, slow down because agents plan on stale
common knowledge. Ignoring the teammates' private evidence in the value costs no
significant time; an oracle ToM model that knows this evidence was 3.1 [$-4.3$, 10.7] slots
faster in the default setting. Over the team sizes
$K\in\{2,\dots,10\}$ and budgets $W\in\{25,\dots,400\}$, GOSC is faster than semantic periodic
updates for $K=8$ and $K=10$ (6.9 [1.3, 12.9] and 8.7 [3.7, 14.1] slots) and at $W=50$ (17.8
[7.0, 28.5]), and slower for $K=2$ (24.9 slots) and for $W\ge200$ (7--8 slots); with two
agents or abundant resources, reporting everything is cheap and helps.

\subsection{Proof of Concept with Heterogeneous LLM--Planner Teams}\label{sec:llm}
This experiment is a proof of concept, not a study of agentic reasoning; it tests whether
GOSC's communication layer works unchanged when the cognitive core of some agents is an
LLM. Agents 0, 2 and 4 are driven by Qwen2.5-3B-Instruct, Llama-3.2-3B-Instruct or
Qwen2.5-7B-Instruct (4-bit, served locally). The LLM receives five pre-processed candidate
targets with their survivor mass, distance and distance to the nearest teammate, plus the
teammates' announced targets, and returns a label constrained by a JSON schema; this constrained choice is far from
open-ended agentic reasoning. Responses are cached, so every mission is reproducible. The
answers were always valid, but the LLMs chose the planner's candidate in only
\LLMAgreeLo--\LLMAgreeHi\% of the deliberations, and the LLM teams were weaker searchers
than the planner team even with ideal communication (\LLMGenieGap{} slots slower with
Qwen2.5-3B).

Table~\ref{tab:llm} compares GOSC with semantic periodic updates in the search task. Two
results must be distinguished. \emph{Resource use} transfers unchanged across cognitive
cores; GOSC used \LLMRatioRange{} times fewer channel uses with every model, and
3.0 times fewer in the planner team. \emph{Completion time} depends on the team and is
less certain. With Qwen2.5-3B (50 seeds), GOSC was faster by \LLMdTQwen{} slots. With
Llama-3.2-3B, the difference was not significant on the 30 seeds of our first run (17.3
[$-0.2$, 35.6] slots); on 50 seeds it was \LLMdTLlama{} slots, significant but with a
lower bound close to zero. \QSevenText{} In the planner team, no significant difference was
detected (\LLMdTPlanner{} slots). Where the difference is significant, the periodic
scheme degrades more than GOSC when LLM agents join (per-seed interaction \LLMInterQwen{}
slots with Qwen2.5-3B). In these teams, too, no significant difference between the ToM and
throughput values was detected. We therefore claim only that the communication layer transfers
across cognitive cores while preserving its resource savings; larger models and
richer agentic tasks are needed to establish completion-time effects.

\begin{table}[t]
\centering
\caption{Proof of concept with heterogeneous teams, search task (agents 0, 2, 4
LLM-driven; missions in parentheses). Completion times are means $\pm$ 95\% CI; channel uses
($\times10^3$) are means; $\Delta T$ is the paired difference periodic minus GOSC [95\%
bootstrap CI], positive when GOSC is faster; the ratio is channel uses periodic over GOSC}
\label{tab:llm}
\setlength{\tabcolsep}{3pt}
\resizebox{\columnwidth}{!}{%
\begin{tabular}{@{}llrrlr@{}}
\toprule
Team & Scheme & Completion & Ch.\ uses & $\Delta T$ (slots) & Ratio\\
\midrule
\csname @@input\endcsname figures/table_llm_r2.tex
\bottomrule
\end{tabular}}
\end{table}

\subsection{Limitations}\label{sec:limits}
The results cover two grid tasks in one simulator with a generic channel model; the
efficiency claims are relative to the periodic policies we tuned and to targets set from
an ideal-communication reference obtained with the same heuristic planner. We did not
reimplement the published systems of \cite{wang2022tom2c,thomas2023pragmatic,wu2024waypoint},
which are trained or not publicly available, and evaluated their mechanisms instead. The
ablations change resource use as well as outcome, so they do not isolate joint message
and rate selection from an equally tuned sequential scheduler at matched cost. The LLM
experiment is a proof of concept with small models and a constrained decision
(Section~\ref{sec:llm}).

\section{Conclusion}\label{sec:conclusion}
We proposed GOSC, a closed-loop co-design of what an agent sends, when it sends it, and
how reliably it is delivered, with the broadcast of common knowledge closing the loop. We
proved that a message is sent only if its value exceeds the price of its delivery, that
more valuable messages receive more robust rates, and that the Lagrangian scheduler is
exact whenever the price keeps the budget slack, which held in 95--96\% of the logged
instances. At operating points frozen and re-evaluated on fresh seeds, GOSC met common
mission targets with 1.2--8.5 times fewer uplink channel uses than tuned
semantic periodic schedules that filter or suppress updates, and it saved significantly
more survivors in two of three rescue settings. The advantage survived 16--64-bit packet
overheads in five of six settings, held for the total uplink-plus-downlink cost, and
persisted in the rescue task when all agents shared one uplink. Precise semantic
valuation was not needed; rough value estimates preserved the advantage in most settings,
whereas content-independent values lost it in the search task. With three LLMs, the
resource savings carried over, while completion-time gains depended on the model. Tasks
in which individual messages have concentrated, predictable consequences, richer agentic
reasoning with larger LLMs, and multi-cell radio resource management are natural next
steps.

\appendix[Implementation Details]\label{app}
\emph{Rescue policy and claim value.} For a known task $s$ with deadline $d_s$ and service
progress $\pi_s$, agent $k$ at $p_k$ computes the slack $d_s-t-\|p_k-s\|_1-(S_r-\pi_s)$. It
keeps its rescue target while the slack is non-negative and no higher-priority teammate
has announced the same target; otherwise it takes the feasible task with the smallest
slack that no teammate is known to claim and for which no free teammate at a known
position is closer, and if none exists it searches with \eqref{eq:score}. Announcing a
rescue target $s$ has value $v_C=5$ if, applying this policy to the shared knowledge,
some teammate would choose $s$, and $0.05$ otherwise.

\emph{Utility in \eqref{eq:voi}.} $U_j(c)=M^{(k)}(c)\,(1+\beta\|c-p_j\|_1)^{-1}
\prod_{q\in\mathcal{D}_j^{(k)}}\rho(c,q)$, where $M^{(k)}$ is the footprint mass under the
sender's belief $\ell^E+\delta_k$ (excluding declared cells and the sender's own confirmed
survivors) and $\mathcal{D}_j^{(k)}$ is teammate $j$'s deconfliction set with the sender's
point replaced by its intention $g_k$. For evidence, $\mathcal{K}^E\oplus m$ adds the
records' LLRs to $\ell^E$ and updates the sender's known position; for an intent, it
replaces the sender's point in the deconfliction sets of lower-priority teammates.

\emph{Relevance filter and suppression.} A record $o$ is retained if
$\sum_{c\in\mathcal{F}(o)\setminus\mathcal{X}}|\sigma(\ell^E+q_k+\lambda(o))(c)-\sigma(\ell^E+q_k)(c)|\ge\varepsilon$,
where $q_k$ is the LLR of the records already queued and $\mathcal{X}$ the declared cells;
GOSC uses $\varepsilon=0.02$, and the suppressing periodic baseline the swept $\theta$.

\emph{Finite blocklength.} An $n$-symbol packet at rate $R$ and SNR $\gamma$ fails with
$\epsilon=Q\big((n[C(\gamma)-R]+\tfrac12\log_2 n)/\sqrt{nV(\gamma)}\big)$,
$C=\log_2(1+\gamma)$, $V=(1-(1+\gamma)^{-2})\log_2^2e$; a coding gap $\Delta$ scales $\gamma$ by
$10^{-\Delta/10}$, and the transmitter uses the fading average $\E_h[1-\epsilon]$.

\emph{Overhead and cost accounting.} Every GOSC packet and every per-slot transport block
of a periodic scheme (which multiplexes its queue) carries $H$ extra bits, so GOSC pays one
overhead per message and the periodic schemes one per slot, which favours the latter; an
option whose packet and overhead do not fit into the slot is unavailable at that rate,
and the periodic schemes use the goodput-optimal rate for $(WR-H)^+$ payload bits. Every
non-empty broadcast carries $H$ as well.

\emph{Imperfect values.} For a log-normal error, the value of each intent and evidence
message is multiplied by $e^{\sigma z}$, $z\sim\mathcal{N}(0,1)$, drawn per message and slot
and common to its options (the setting of Proposition~\ref{prop:robust}); the
order-of-magnitude estimate rounds each value to the nearest power of ten; the
content-independent value replaces the magnitude by that of a randomly drawn earlier
message of the same type in the mission (keeping its distribution but not its
information). Confirmations keep their fixed value; $\rho$ is the Spearman correlation
between estimated and true values, pooled over missions.

\emph{Frozen operating points.} Table~\ref{tab:frozen} lists the operating points behind
Table~\ref{tab:val}; all other parameters are as in Table~\ref{tab:params}.

\begin{table}[h]
\centering
\caption{Frozen operating points of the fresh-seed validation (Table~\ref{tab:val})}
\label{tab:frozen}
{\footnotesize
\begin{tabular}{@{}lll@{}}
\toprule
Setting & GOSC price & Periodic configuration\\
\midrule
\csname @@input\endcsname figures/table_frozen.tex
\bottomrule
\end{tabular}}
\end{table}

\bibliographystyle{IEEEtran}
\bibliography{refs}

@book{shannon1949,
  author    = {C. E. Shannon and W. Weaver},
  title     = {The Mathematical Theory of Communication},
  publisher = {University of Illinois Press},
  address   = {Urbana, IL},
  year      = {1949}
}

@article{gunduz2023beyond,
  author  = {D. G{\"u}nd{\"u}z and Z. Qin and I. E. Aguerri and H. S. Dhillon and Z. Yang and A. Yener and K. K. Wong and C.-B. Chae},
  title   = {Beyond Transmitting Bits: Context, Semantics, and Task-Oriented Communications},
  journal = {IEEE J. Sel. Areas Commun.},
  volume  = {41},
  number  = {1},
  pages   = {5--41},
  year    = {2023}
}

@article{strinati2021beyond,
  author  = {E. Calvanese Strinati and S. Barbarossa},
  title   = {{6G} Networks: Beyond {S}hannon towards Semantic and Goal-Oriented Communications},
  journal = {Comput. Netw.},
  volume  = {190},
  pages   = {107930},
  year    = {2021}
}

@article{kountouris2021semantics,
  author  = {M. Kountouris and N. Pappas},
  title   = {Semantics-Empowered Communication for Networked Intelligent Systems},
  journal = {IEEE Commun. Mag.},
  volume  = {59},
  number  = {6},
  pages   = {96--102},
  year    = {2021}
}

@article{chaccour2025less,
  author  = {C. Chaccour and W. Saad and M. Debbah and Z. Han and H. V. Poor},
  title   = {Less Data, More Knowledge: Building Next-Generation Semantic Communication Networks},
  journal = {IEEE Commun. Surveys Tuts.},
  volume  = {27},
  number  = {1},
  year    = {2025}
}

@article{bourtsoulatze2019deep,
  author  = {E. Bourtsoulatze and D. {Burth Kurka} and D. G{\"u}nd{\"u}z},
  title   = {Deep Joint Source-Channel Coding for Wireless Image Transmission},
  journal = {IEEE Trans. Cogn. Commun. Netw.},
  volume  = {5},
  number  = {3},
  pages   = {567--579},
  year    = {2019}
}

@article{xie2021deepsc,
  author  = {H. Xie and Z. Qin and G. Y. Li and B.-H. Juang},
  title   = {Deep Learning Enabled Semantic Communication Systems},
  journal = {IEEE Trans. Signal Process.},
  volume  = {69},
  pages   = {2663--2675},
  year    = {2021}
}

@article{xie2022taskmulti,
  author  = {H. Xie and Z. Qin and X. Tao and K. B. Letaief},
  title   = {Task-Oriented Multi-User Semantic Communications},
  journal = {IEEE J. Sel. Areas Commun.},
  volume  = {40},
  number  = {9},
  pages   = {2584--2597},
  year    = {2022}
}

@article{shao2022ib,
  author  = {J. Shao and Y. Mao and J. Zhang},
  title   = {Learning Task-Oriented Communication for Edge Inference: An Information Bottleneck Approach},
  journal = {IEEE J. Sel. Areas Commun.},
  volume  = {40},
  number  = {1},
  pages   = {197--211},
  year    = {2022}
}

@article{tung2021effective,
  author  = {T.-Y. Tung and S. Kobus and J. P. Roig and D. G{\"u}nd{\"u}z},
  title   = {Effective Communications: A Joint Learning and Communication Framework for Multi-Agent Reinforcement Learning Over Noisy Channels},
  journal = {IEEE J. Sel. Areas Commun.},
  volume  = {39},
  number  = {8},
  pages   = {2590--2603},
  year    = {2021}
}

@article{uysal2022semantic,
  author  = {E. Uysal and O. Kaya and A. Ephremides and J. Gross and M. Codreanu and P. Popovski and M. Assaad and G. Liva and A. Munari and B. Soret and T. Soleymani and K. H. Johansson},
  title   = {Semantic Communications in Networked Systems: A Data Significance Perspective},
  journal = {IEEE Netw.},
  volume  = {36},
  number  = {4},
  pages   = {233--240},
  year    = {2022}
}

@article{saad2020vision,
  author  = {W. Saad and M. Bennis and M. Chen},
  title   = {A Vision of {6G} Wireless Systems: Applications, Trends, Technologies, and Open Research Problems},
  journal = {IEEE Netw.},
  volume  = {34},
  number  = {3},
  pages   = {134--142},
  year    = {2020}
}

@article{letaief2019roadmap,
  author  = {K. B. Letaief and W. Chen and Y. Shi and J. Zhang and Y.-J. A. Zhang},
  title   = {The Roadmap to {6G}: {AI} Empowered Wireless Networks},
  journal = {IEEE Commun. Mag.},
  volume  = {57},
  number  = {8},
  pages   = {84--90},
  year    = {2019}
}

@inproceedings{foerster2016learning,
  author    = {J. Foerster and I. A. Assael and N. de Freitas and S. Whiteson},
  title     = {Learning to Communicate with Deep Multi-Agent Reinforcement Learning},
  booktitle = {Proc. Adv. Neural Inf. Process. Syst. (NeurIPS)},
  pages     = {2137--2145},
  year      = {2016}
}

@inproceedings{sukhbaatar2016commnet,
  author    = {S. Sukhbaatar and A. Szlam and R. Fergus},
  title     = {Learning Multiagent Communication with Backpropagation},
  booktitle = {Proc. Adv. Neural Inf. Process. Syst. (NeurIPS)},
  pages     = {2244--2252},
  year      = {2016}
}

@inproceedings{das2019tarmac,
  author    = {A. Das and T. Gervet and J. Romoff and D. Batra and D. Parikh and M. Rabbat and J. Pineau},
  title     = {{TarMAC}: Targeted Multi-Agent Communication},
  booktitle = {Proc. Int. Conf. Mach. Learn. (ICML)},
  pages     = {1538--1546},
  year      = {2019}
}

@inproceedings{kim2019schednet,
  author    = {D. Kim and S. Moon and D. Hostallero and W. J. Kang and T. Lee and K. Son and Y. Yi},
  title     = {Learning to Schedule Communication in Multi-Agent Reinforcement Learning},
  booktitle = {Proc. Int. Conf. Learn. Represent. (ICLR)},
  year      = {2019}
}

@article{bernstein2002complexity,
  author  = {D. S. Bernstein and R. Givan and N. Immerman and S. Zilberstein},
  title   = {The Complexity of Decentralized Control of {M}arkov Decision Processes},
  journal = {Math. Oper. Res.},
  volume  = {27},
  number  = {4},
  pages   = {819--840},
  year    = {2002}
}

@inproceedings{kaul2012aoi,
  author    = {S. Kaul and R. Yates and M. Gruteser},
  title     = {Real-Time Status: How Often Should One Update?},
  booktitle = {Proc. IEEE INFOCOM},
  pages     = {2731--2735},
  year      = {2012}
}

@article{maatouk2020aoii,
  author  = {A. Maatouk and S. Kriouile and M. Assaad and A. Ephremides},
  title   = {The Age of Incorrect Information: A New Performance Metric for Status Updates},
  journal = {IEEE/ACM Trans. Netw.},
  volume  = {28},
  number  = {5},
  pages   = {2215--2228},
  year    = {2020}
}

@inproceedings{rao1995bdi,
  author    = {A. S. Rao and M. P. Georgeff},
  title     = {{BDI} Agents: From Theory to Practice},
  booktitle = {Proc. 1st Int. Conf. Multi-Agent Syst. (ICMAS)},
  pages     = {312--319},
  year      = {1995}
}

@article{baker2017tom,
  author  = {C. L. Baker and J. Jara-Ettinger and R. Saxe and J. B. Tenenbaum},
  title   = {Rational Quantitative Attribution of Beliefs, Desires and Percepts in Human Mentalizing},
  journal = {Nat. Hum. Behav.},
  volume  = {1},
  pages   = {0064},
  year    = {2017}
}

@inproceedings{park2023generative,
  author    = {J. S. Park and J. O'Brien and C. J. Cai and M. R. Morris and P. Liang and M. S. Bernstein},
  title     = {Generative Agents: Interactive Simulacra of Human Behavior},
  booktitle = {Proc. ACM Symp. User Interface Softw. Technol. (UIST)},
  year      = {2023}
}

@inproceedings{li2023camel,
  author    = {G. Li and H. A. A. K. Hammoud and H. Itani and D. Khizbullin and B. Ghanem},
  title     = {{CAMEL}: Communicative Agents for ``Mind'' Exploration of Large Language Model Society},
  booktitle = {Proc. Adv. Neural Inf. Process. Syst. (NeurIPS)},
  year      = {2023}
}

@article{wu2023autogen,
  author  = {Q. Wu and G. Bansal and J. Zhang and Y. Wu and B. Li and E. Zhu and L. Jiang and X. Zhang and S. Zhang and J. Liu and A. H. Awadallah and R. W. White and D. Burger and C. Wang},
  title   = {{AutoGen}: Enabling Next-Gen {LLM} Applications via Multi-Agent Conversation},
  journal = {arXiv preprint arXiv:2308.08155},
  year    = {2023}
}

@inproceedings{hong2024metagpt,
  author    = {S. Hong and M. Zhuge and J. Chen and X. Zheng and Y. Cheng and C. Zhang and J. Wang and Z. Wang and S. K. S. Yau and Z. Lin and L. Zhou and C. Ran and L. Xiao and C. Wu and J. Schmidhuber},
  title     = {{MetaGPT}: Meta Programming for a Multi-Agent Collaborative Framework},
  booktitle = {Proc. Int. Conf. Learn. Represent. (ICLR)},
  year      = {2024}
}

@inproceedings{zhang2024coela,
  author    = {H. Zhang and W. Du and J. Shan and Q. Zhou and Y. Du and J. B. Tenenbaum and T. Shu and C. Gan},
  title     = {Building Cooperative Embodied Agents Modularly with Large Language Models},
  booktitle = {Proc. Int. Conf. Learn. Represent. (ICLR)},
  year      = {2024}
}

@book{tse2005fundamentals,
  author    = {D. Tse and P. Viswanath},
  title     = {Fundamentals of Wireless Communication},
  publisher = {Cambridge University Press},
  year      = {2005}
}

@book{boyd2004convex,
  author    = {S. Boyd and L. Vandenberghe},
  title     = {Convex Optimization},
  publisher = {Cambridge University Press},
  year      = {2004}
}

@article{premack1978tom,
  author  = {D. Premack and G. Woodruff},
  title   = {Does the Chimpanzee Have a Theory of Mind?},
  journal = {Behav. Brain Sci.},
  volume  = {1},
  number  = {4},
  pages   = {515--526},
  year    = {1978}
}

@inproceedings{wang2022tom2c,
  author    = {Y. Wang and F. Zhong and J. Xu and Y. Wang},
  title     = {{ToM2C}: Target-Oriented Multi-Agent Communication and Cooperation with Theory of Mind},
  booktitle = {Proc. Int. Conf. Learn. Represent. (ICLR)},
  year      = {2022}
}

@article{thomas2023pragmatic,
  author  = {C. K. Thomas and E. Calvanese Strinati and W. Saad},
  title   = {Reasoning with the Theory of Mind for Pragmatic Semantic Communication},
  journal = {arXiv preprint arXiv:2311.18224},
  year    = {2023}
}

@misc{3gpp38211,
  author       = {{3GPP}},
  title        = {{NR}; Physical Channels and Modulation},
  howpublished = {3rd Generation Partnership Project (3GPP), Technical Specification (TS) 38.211}
}

@article{polyanskiy2010channel,
  author  = {Y. Polyanskiy and H. V. Poor and S. Verd{\'u}},
  title   = {Channel Coding Rate in the Finite Blocklength Regime},
  journal = {IEEE Trans. Inf. Theory},
  volume  = {56},
  number  = {5},
  pages   = {2307--2359},
  year    = {2010}
}

@article{durisi2016toward,
  author  = {G. Durisi and T. Koch and P. Popovski},
  title   = {Toward Massive, Ultrareliable, and Low-Latency Wireless Communication with Short Packets},
  journal = {Proc. IEEE},
  volume  = {104},
  number  = {9},
  pages   = {1711--1726},
  year    = {2016}
}

@book{kellerer2004knapsack,
  author    = {H. Kellerer and U. Pferschy and D. Pisinger},
  title     = {Knapsack Problems},
  publisher = {Springer},
  year      = {2004}
}

@article{kadota2018scheduling,
  author  = {I. Kadota and A. Sinha and E. Uysal-Biyikoglu and R. Singh and E. Modiano},
  title   = {Scheduling Policies for Minimizing Age of Information in Broadcast Wireless Networks},
  journal = {IEEE/ACM Trans. Netw.},
  volume  = {26},
  number  = {6},
  pages   = {2637--2650},
  year    = {2018}
}

@article{holm2023goal,
  author  = {J. Holm and F. Chiariotti and A. E. Kal{\o}r and B. Soret and T. B. Pedersen and P. Popovski},
  title   = {Goal-Oriented Scheduling in Sensor Networks with Application Timing Awareness},
  journal = {IEEE Trans. Commun.},
  year    = {2023}
}

@article{chiariotti2022qaoi,
  author  = {F. Chiariotti and J. Holm and A. E. Kal{\o}r and B. Soret and S. K. Jensen and T. B. Pedersen and P. Popovski},
  title   = {Query Age of Information: Freshness in Pull-Based Communication},
  journal = {IEEE Trans. Commun.},
  volume  = {70},
  number  = {3},
  year    = {2022}
}

@article{wu2024waypoint,
  author  = {W. Wu and Y. Yang and Y. Deng and A. H. Aghvami},
  title   = {Goal-Oriented Semantic Communications for Robotic Waypoint Transmission: The Value and Age of Information Approach},
  journal = {IEEE Trans. Wireless Commun.},
  year    = {2024}
}

@article{soleymani2022voi,
  author  = {T. Soleymani and J. S. Baras and S. Hirche},
  title   = {Value of Information in Feedback Control: Quantification},
  journal = {IEEE Trans. Autom. Control},
  volume  = {67},
  number  = {7},
  pages   = {3730--3737},
  year    = {2022}
}

@article{chou2006rd,
  author  = {P. A. Chou and Z. Miao},
  title   = {Rate-Distortion Optimized Streaming of Packetized Media},
  journal = {IEEE Trans. Multimedia},
  volume  = {8},
  number  = {2},
  pages   = {390--404},
  year    = {2006}
}

@article{sun2025importance,
  author  = {Z. Sun and S. Ma and S. Li},
  title   = {Task-Oriented Semantic Communication With Importance-Aware Rate Control},
  journal = {IEEE Commun. Lett.},
  volume  = {29},
  number  = {7},
  pages   = {1520--1524},
  year    = {2025},
  doi     = {10.1109/LCOMM.2025.3566092}
}

@article{cao2025mimo,
  author  = {Y. Cao and Y. Wu and L. Lian and M. Tao},
  title   = {Importance-Aware Resource Allocations for {MIMO} Semantic Communication},
  journal = {Entropy},
  volume  = {27},
  number  = {6},
  pages   = {605},
  year    = {2025},
  doi     = {10.3390/e27060605}
}

@article{kim2025uep,
  author  = {S. Kim and Y. Oh and Y. Kim and N. Lee and Y.-S. Jeon},
  title   = {Unequal Error Protection for Digital Semantic Communication with Channel Coding},
  journal = {arXiv preprint arXiv:2508.03381},
  year    = {2025}
}

@article{lei2026icotasc,
  author  = {K. Lei and Y. Peng and L. Zhang and J. Xu},
  title   = {Importance-aware Resource Allocation for Collaborative Task-Oriented Semantic Communication},
  journal = {arXiv preprint arXiv:2606.29052},
  year    = {2026}
}

@article{everett1963generalized,
  author  = {H. Everett},
  title   = {Generalized {L}agrange Multiplier Method for Solving Problems of Optimum Allocation of Resources},
  journal = {Oper. Res.},
  volume  = {11},
  number  = {3},
  pages   = {399--417},
  year    = {1963}
}

@book{topkis1998supermodularity,
  author    = {D. M. Topkis},
  title     = {Supermodularity and Complementarity},
  publisher = {Princeton Univ. Press},
  year      = {1998}
}

@inproceedings{astrom2002riemann,
  author    = {K. J. {\AA}str{\"o}m and B. Bernhardsson},
  title     = {Comparison of {R}iemann and {L}ebesgue Sampling for First Order Stochastic Systems},
  booktitle = {Proc. IEEE Conf. Decis. Control (CDC)},
  pages     = {2011--2016},
  year      = {2002}
}

@article{kelly1998rate,
  author  = {F. P. Kelly and A. K. Maulloo and D. K. H. Tan},
  title   = {Rate Control for Communication Networks: Shadow Prices, Proportional Fairness and Stability},
  journal = {J. Oper. Res. Soc.},
  volume  = {49},
  number  = {3},
  pages   = {237--252},
  year    = {1998}
}

@misc{3gpp38212,
  author       = {{3GPP}},
  title        = {{NR}; Multiplexing and Channel Coding},
  howpublished = {3rd Generation Partnership Project (3GPP), Technical Specification (TS) 38.212}
}

@misc{3gpp38321,
  author       = {{3GPP}},
  title        = {{NR}; Medium Access Control ({MAC}) Protocol Specification},
  howpublished = {3rd Generation Partnership Project (3GPP), Technical Specification (TS) 38.321}
}

@inproceedings{rezazadeh2024genonet,
  author    = {Farhad Rezazadeh and Amir Ashtari Gargari and Sandra Lag{\'e}n
               and Josep Mangues-Bafalluy and Dusit Niyato and Lingjia Liu},
  title     = {{GenOnet}: Generative Open {xG} Network Simulation with
               Multi-Agent {LLM} and {ns-3}},
  booktitle = {2024 3rd International Conference on 6G Networking (6GNet)},
  publisher = {IEEE},
  year      = {2024},
  doi       = {10.1109/6GNet63182.2024.10765766}
}

\end{document}